\documentclass[11.5pt]{scrartcl}

\usepackage{xy-format}
\usepackage{xy-theorem}
\usepackage{xy-math}
\usepackage{fancyhdr}
\usepackage{longtable}

\title{Chiral Spin(7) Sigma Models and Topological Modular Forms}

\author[1,2]{Xingyang Yu}
\affil[1]{FirstPrinciples, Toronto, Ontario, Canada}
\affil[2]{Department of Physics, Virginia Tech, Blacksburg, Virginia 24061, USA}

\date{}

\begin{document}
\maketitle

\begin{abstract}
We construct canonical chiral \(\mathcal N=(0,1)\) sigma models on Spin(7) manifolds. The Spin(7) structure fixes the rank-seven bundle of the left-moving fermions, whose quadratic index matches the tangent representation, so the internal anomalies cancel while the rank difference leaves the gravitational anomaly of one right-moving Majorana--Weyl fermion. A transverse section of the left-moving bundle has a 1D zero locus with induced String structure, and we compute its class in the first torsion group of topological modular forms. The class is nonzero precisely when the Euler characteristic of the target is odd. The standard Joyce action on the torus has no internal finite-group anomaly and can be gauged to define an exact chiral spin QFT. We show that the orbifold Euler characteristic is even throughout the diagonal Joyce family generated by coordinate reflections and half-shifts preserving the Cayley form. By contrast, a free antiholomorphic quotient of the Fermat sextic is a compact torsion-free Spin(7) target with Euler characteristic 1305 and realizes the nonzero class.
\end{abstract}

\newpage
\tableofcontents
\clearpage

\section{Introduction}
\label{sec:introduction}

Compact Spin(7) manifolds realize the last entry of Berger's holonomy classification, and string compactification on them preserves a minimal amount of supersymmetry \cite{Joyce1996Spin7,ShatashviliVafa1994}. Spin(7) holonomy leaves two familiar signatures on the worldsheet. The Cayley form generates exceptional chiral symmetries and the Ising subsector of the Spin(7) sigma model \cite{Howe:1991ic,ShatashviliVafa1994}. The heterotic standard embedding reveals a second structure. Under \(\mathfrak e_8\supset\mathfrak{so}(7)\oplus\mathfrak{so}(9)\), seven gauge fermions couple to the Spin(7) connection while nine complete the \(SO(9)\) sector \cite{Sugiyama:2001qh}. These observations isolate a natural asymmetric system with eight tangent fermions and seven Spin(7) fermions.

An asymmetric fermion system of this kind defines a 2D \(\mathcal N=(0,1)\) quantum field theory. These theories carry a single Majorana--Weyl supercharge, which is the least supersymmetry in the lowest dimension that admits a chiral one. Spin(7) geometry already engineers such theories directly, through D1-branes probing Spin(7) cones obtained from Calabi--Yau fourfolds by antiholomorphic involutions \cite{Franco:2021ixh,Franco:2021vxq}. Stolz and Teichner conjectured that their deformation classes are classified by topological modular forms \cite{StolzTeichner2004,StolzTeichner2011}, and this correspondence now organizes both worldsheet constructions and heterotic anomaly questions \cite{GukovPeiPutrovVafa2018,TachikawaYamashita2021}. The torsion classes of \(\tmf\) form the subtlest part of the correspondence. Characteristic forms and rational genera cannot detect them, and only mod-two indices see them \cite{TachikawaYamashitaYonekura2023,BerwickEvans2023}. The first torsion class appears in degree one. Which compact target geometry selects it?

In this work we answer this question through the seven-fermion system isolated above. We take it as the intrinsic left-moving content of a chiral \(\mathcal N=(0,1)\) sigma model. The Spin(7) structure fixes its Fermi bundle, cancels the internal anomaly against the tangent fermions, and produces a degree-one String zero-locus class. We compute this class in \(\tmf_1\) and show that it is nonzero precisely when \(\chi(X)\) is odd.

Let \(P\to X\) be a fixed Spin(7) reduction. The tangent bundle is associated to the real spin representation \(\Delta_8\),
\begin{equation}
TX=P\times_{\Spin(7)}\Delta_8.
\end{equation}
The same principal bundle also defines the associated bundle
\begin{equation}
V_7=P\times_{\Spin(7)}\rho_7
\cong \Lambda^2_7T^*X,
\label{eq:intro-v7}
\end{equation}
where \(\rho_7\) is the vector representation of \(\Spin(7)\). We couple the right-moving fermions to \(TX\) and the left-moving fermions to \(V_7\). This is the canonical Spin(7) sigma model studied in this paper.

The key representation-theoretic fact is that \(\Delta_8\) and \(\rho_7\) have the same quadratic index. Their associated bundles obey
\begin{equation}
w_1(TX)=w_1(V_7)=0,
\qquad
w_2(TX)=w_2(V_7)=0,
\qquad
\lambda(TX)=\lambda(V_7),
\label{eq:intro-characteristic-equality}
\end{equation}
with \(\lambda=p_1/2\). Equation~\eqref{eq:intro-characteristic-equality} cancels the internal target and gauge anomaly of the chiral fermions. The remaining rank difference is one. In free-field normalization,
\begin{equation}
c_R=8+\frac{8}{2}=12,
\qquad
c_L=8+\frac{7}{2}=\frac{23}{2},
\qquad
\nu=2(c_R-c_L)=1.
\label{eq:intro-degree}
\end{equation}
The model is therefore a spin QFT with gravitational degree \(\nu=1\) and the invertible anomaly of one chiral Majorana fermion. This anomaly is RG invariant, so the model cannot flow as a whole to the ordinary balanced \(\mathcal N=(1,1)\) sigma model.

The gravitational anomaly fixes the grading, but it does not select the element inside that degree. In degree one \cite{TachikawaYamashita2021},
\begin{equation}
\tmf_1\cong\ZZ/2\{\eta\}.
\end{equation}
Here \(\eta\) denotes the stable Hopf element, not the Dedekind eta function. We use \(\tmf\), \(\Tmf\), and \(\TMF\) for the connective, compactified, and periodic theories, respectively.\footnote{The capitalization distinguishes the three spectra and does not distinguish SCFTs from general QFTs.} The associated topological class is either \(0\) or \(\eta\). We determine which class is selected by \(X\) through geometric localization. Choose a transverse section \(s\) of \(V_7\). Its zero locus
\begin{equation}
Z=s^{-1}(0)
\end{equation}
is a closed 1D manifold, and the stable identity
\begin{equation}
TZ\simeq (TX-V_7)|_Z
\label{eq:intro-stable-zero-locus}
\end{equation}
equips it with the String structure induced from the virtual bundle \(TX-V_7\). The String orientation then defines a class in \(\tmf_1\) \cite{AndoHopkinsRezk2010,AndoEtAl2014}. We compute this class by stabilizing the odd-rank Atiyah--Bott--Shapiro symbol and applying graded Clifford Morita equivalence. For the canonical Spin(7) bundle, the calculation reduces to the parity of the twisted Dirac index \(\ind(D_X\otimes TX_{\mathbb C})\). The result is
\begin{equation}
\boxed{
\mu_7(X)=[Z]_{\tmf}
=\bigl(\chi(X)\bmod2\bigr)\eta.
}
\label{eq:intro-main-result}
\end{equation}
Thus the Euler characteristic selects the String zero-locus class. Both values belong to the same degree-one sector, and \(\mu_7(X)=0\) does not remove the gravitational anomaly.

Equation~\eqref{eq:intro-main-result} is not only a zero-mode formula. The real-\(KO\) Euler class localizes the complete normalized Dirac--Ramond series to \(Z\). Since \(TZ-\RR\) is stably trivial, every positive oscillator contribution cancels and
\begin{equation}
\Phi_2(X,V_7;q)
=\bigl(\chi(X)\bmod2\bigr)\eta
\end{equation}
as an entire \(q\)-series. In degree one, the cusp map for \(KO_{\mathrm{MF}}\) is an isomorphism, so the same series determines the full class-level lift.\footnote{This is the normalized real-\(KO\) index rather than the complete torus partition function.}

The free Joyce point provides the most concrete worldsheet test. The eight right-moving fermions transform in \(\Delta_8\), while the seven left-moving fermions transform in \(\rho_7\). We compute the internal finite-spin anomaly for the Joyce group \(\Gamma=(\ZZ_2)^4\) and find that it vanishes. After choosing a trivialization of this internal anomaly, ordinary finite gauging defines a gravitationally anomalous spin orbifold CFT. The gauged model retains the single-Majorana gravitational anomaly in \eqref{eq:intro-degree}. Its unweighted orbifold Euler characteristic is \(144\), matching the even Euler characteristic and vanishing geometric Smith charge of the corresponding smooth Joyce target.

The same Joyce action appears in the genus-one modular-invariant heterotic construction of Sugiyama and Yamaguchi \cite{Sugiyama:2001qh}. Its gauge sector realizes the conformal embedding
\begin{equation}
E_{8,1}\supset \mathfrak{so}(7)_1\oplus\mathfrak{so}(9)_1.
\end{equation}
This embedding does not give a tensor-product decomposition, so the standalone seven-fermion theory cannot be recovered by dividing characters. The ordinary \(\mathcal N=(1,1)\) Joyce orbifold gives a second comparison theory with balanced tangent fermions and no gravitational anomaly.

The standard Joyce sign vectors generate the binary subspace \(C\subset\FF_2^8\) that gives the complete diagonal stabilizer of the Cayley form. We show that every effective elementary abelian half-affine subgroup of
\begin{equation}
C\oplus\FF_2^8
\end{equation}
has even unweighted orbifold Euler characteristic. For full-sign lifts, the only values are
\begin{equation}
48,\qquad 144,\qquad 336,\qquad 720.
\end{equation}
The theorem covers arbitrary pure half-translation kernels and all quarter-period conjugacy classes. Two larger finite computations find the same parity throughout the signed-monomial families they enumerate. The standard Joyce point is therefore one member of a much larger even family.

This evenness is not universal in torsion-free Spin(7) geometry. Let
\begin{equation}
X_{\mathrm F}=\{z_0^6+\cdots+z_5^6=0\}\subset\mathbb{CP}^5.
\end{equation}
Complex conjugation acts freely on \(X_{\mathrm F}\), and the quotient \(Y_{\mathrm F}=X_{\mathrm F}/\ZZ_2\) is a compact torsion-free barely Spin(7) manifold with
\begin{equation}
\pi_1(Y_{\mathrm F})=\ZZ_2,
\qquad
\Hol(Y_{\mathrm F})=SU(4)\rtimes\ZZ_2,
\qquad
\chi(Y_{\mathrm F})=1305.
\end{equation}
It follows that
\begin{equation}
\mu_7(Y_{\mathrm F})=\eta.
\end{equation}
The target itself is not String. The relevant String structure belongs to the virtual pair \(TY_{\mathrm F}-V_7\). The left-moving fermions cancel the obstruction that prevents \(TY_{\mathrm F}\) from carrying the required orientation.

The rest of the paper is organized as follows. Section~\ref{sec:canonical-model} constructs the sigma-model pair and derives its universal anomaly cancellation. Section~\ref{sec:spin7-smith-charge} follows one calculation from the Fermi zero locus to Euler parity, oscillator cancellation, and the degree-one \(KO_{\mathrm{MF}}\) class. Section~\ref{sec:joyce-orbifolds} develops the Joyce worldsheet and proves the evenness of the diagonal family. Section~\ref{sec:fermat-sextic} realizes the nonzero class outside that family. Section~\ref{sec:discussion} collects the remaining field-theoretic and geometric problems.

\section{The canonical Spin(7) sigma model}
\label{sec:canonical-model}

We now construct the canonical chiral \(\mathcal N=(0,1)\) sigma model on a Spin(7) manifold \(X\). We first identify the bundles that carry the eight right-moving fermions and the seven left-moving fermions. We then show that their target-frame anomalies cancel because the two Spin(7) representations have the same quadratic index. Their rank difference leaves the gravitational anomaly of a single right-moving Majorana--Weyl fermion.

\subsection{The chiral fields and the two Spin(7) bundles}
\label{subsec:sigma-model-anomaly}

Let \(\Sigma\) be a spin worldsheet, let \(X\) be an 8D Riemannian manifold, and let \(E\to X\) be a real Euclidean bundle coupled to the left-moving fermions. This chiral field content goes back to the early heterotic \((1,0)\) sigma-model literature, with left and right exchanged in our convention \cite{Hull:1985jv,Bergshoeff:1985gc,deLaOssa:2018iij}. A scalar multiplet consists locally of a coordinate field \(\phi^i\) and a right-moving Majorana--Weyl fermion \(\psi_+^i\) valued in the tangent space at \(\phi\). A Fermi multiplet contributes an independent left-moving Majorana--Weyl fermion \(\lambda_-^a\) valued in \(E_\phi\). Globally, these fields are
\begin{equation}
\phi\colon\Sigma\longrightarrow X,
\qquad
\psi_+\in\Gamma\bigl(S_R\otimes\phi^*TX\bigr),
\qquad
\lambda_-\in\Gamma\bigl(S_L\otimes\phi^*E\bigr).
\label{eq:01-field-content}
\end{equation}
Here \(\phi\) is the sigma-model map from the worldsheet \(\Sigma\) to the target \(X\). The pullback \(\phi^*TX\) places the tangent space \(T_{\phi(\sigma)}X\) over each point \(\sigma\in\Sigma\), while \(\phi^*E\) similarly places the fiber \(E_{\phi(\sigma)}\) there. The bundles \(S_R\) and \(S_L\) are the right-moving and left-moving spinor bundles on the spin worldsheet. The symbol \(\Gamma\) denotes the space of sections over \(\Sigma\). Thus \(\psi_+(\sigma)\) carries a right-moving worldsheet spinor index and a tangent-space index, while \(\lambda_-(\sigma)\) carries a left-moving worldsheet spinor index and an \(E\)-fiber index.

Since \(X\) is 8D, the scalar sector contains eight scalar multiplets, whose field components include eight right-moving fermions. Supersymmetry thus fixes the right-moving bundle to be \(TX\). On the other hand, the number of Fermi multiplets is independent of the number of scalar multiplets, so \(\mathcal N=(0,1)\) supersymmetry does not fix the rank or representation of \(E\) coupled to the left-moving fermions.

Local components of the fermions on overlapping patches of target space $X$ are related by frame transformations, and the bundle records this patching. These frame transformations are redundancies of the description of the sigma model. They are not an additional global Spin(7) symmetry or a dynamical gauge symmetry on the 2D worldsheet QFT.\footnote{After pullback, the connections act like background gauge fields for local frame transformations. This local flavor-space language does not imply a global flavor symmetry or a dynamical 2D gauge field. Heterotic sigma-model literature often calls \(E\) the gauge bundle, while the frame principal bundle from which \(E\) is associated is a distinct object.}

Given metrics and compatible connections, the action takes the schematic form \cite{Sen:1985eb,Bergshoeff:1985gc,deLaOssa:2018iij}
\begin{equation}
\begin{split}
S=\frac{1}{4\pi\alpha'}\int_\Sigma
\Bigl(&g_{ij}\,\partial_+\phi^i\partial_-\phi^j
+i g_{ij}\psi_+^iD_-\psi_+^j
+i h_{ab}\lambda_-^aD_+\lambda_-^b
\\
&+\frac{1}{2}F_{ijab}\psi_+^i\psi_+^j\lambda_-^a\lambda_-^b
\Bigr).
\end{split}
\label{eq:01-sigma-action}
\end{equation}
The four-fermion interaction is fixed by \(\mathcal N=(0,1)\) supersymmetry. We have suppressed the \(B\)-field coupling, whose differential data completes the quantum anomaly trivialization \cite{Hull:1986xn}. The two fermion bundles determine both the target-frame anomaly and the topological charge studied below.

Quantum mechanically, the path integral over a real chiral fermion is the Pfaffian of its chiral kinetic operator rather than a determinant (see, e.g., \cite{Freed:1999vc}). As \(\phi\) and the background connections vary, the Pfaffians of \(\psi_+\) and \(\lambda_-\) are sections of line bundles over the space of fields and backgrounds. Their product defines a quantum amplitude only after the combined anomaly line has been trivialized \cite{Freed:1986zx}. The corresponding characteristic-class obstructions are \cite{Hull:1986xn,GukovPeiPutrov2020,deLaOssa:2018iij}
\begin{equation}
w_1(TX)=w_1(E),
\qquad
w_2(TX)=w_2(E),
\qquad
\frac{1}{2}\bigl[p_1(TX)-p_1(E)\bigr]=0,
\label{eq:01-anomaly-condition}
\end{equation}
where the first two equations concern the first and second Stiefel--Whitney classes \(w_i(TX),w_i(E)\in H^i(X,\mathbb Z_2)\). If both \(TX\) and \(E\) admit spin lifts, then \(w_1(TX)=w_1(E)=0\) and \(w_2(TX)=w_2(E)=0\), so only the Pontryagin-class condition remains. Together, these conditions remove the characteristic-class obstructions to patching the combined real-fermion Pfaffian. A residual torsion anomaly is a separate question, which will be resolved below for the Spin(7) model we construct.

When both bundles $TX$ and $E$ are spin, the last equation becomes \(\lambda(TX)=\lambda(E)\), with \(\lambda=p_1/2\). This cancellation should not be confused with the gravitational anomaly, whose coefficient only depends on the net number of Majorana-Weyl fermions,
\begin{equation}
\nu=2(c_R-c_L)=\dim X-\operatorname{rk}E,
\label{eq:degree-general}
\end{equation}
where the rank of tangent bundle $TX$ is just the dimension of the target manifold $X$. 

We now specialize to a manifold \(X\) with a specified Spin(7) structure \(P\to X\). Under the embedding \(\Spin(7)\subset SO(8)\), the tangent representation is the real spin representation \(\Delta_8\) of \(\Spin(7)\). The associated tangent bundle is
\begin{equation}
TX=P\times_{\Spin(7)}\Delta_8.
\label{eq:tangent-spin-rep}
\end{equation}
The associated-bundle notation says that a change of local Spin(7) frame is accompanied by the inverse action on the fiber coordinate.\footnote{Explicitly, \((p,v)\sim(pg,\Delta_8(g^{-1})v)\). Here \(p\) is a local Spin(7) frame and \(v\) records the vector components in that frame. Transforming both entries describes the same geometric tangent vector, just as changing a basis and transforming the components leaves a vector unchanged.} The resulting equivalence class is a frame-independent tangent vector. The right-moving fermions therefore transform in \(\Delta_8\) because they are superpartners of the target coordinates $\phi$.

The same principal Spin(7) bundle has a second basic real representation, which is the vector representation. Let \(\rho_7\) denote the vector representation of \(\Spin(7)\) and we can define a rank-7 bundle 
\begin{equation}
V_7=P\times_{\Spin(7)}\rho_7.
\label{eq:v7-associated}
\end{equation}
To make this bundle more explicit, recall that any Spin(7) manifold is equipped with a Cayley four-form \(\Phi\), whose stabilizer at each point is \(\Spin(7)\). The form determines the metric and orientation, is self-dual with \(*_8\Phi=\Phi\), and has comass one. When the Spin(7) structure is torsion-free, \(\nabla\Phi=0\) and \(d\Phi=0\). The 4D submanifolds on which \(\Phi\) restricts to the induced volume form are the Cayley submanifolds \cite{Harvey:1982xk,Joyce2002Exceptional}. Now, the bundle of two-forms $\Lambda^2T^*X$ can be decomposed, via the induced endomorphism \(\alpha\mapsto *(\Phi\wedge\alpha)\), into rank-7 and rank-21 eigenspaces,
\begin{equation}
\Lambda^2T^*X
=\Lambda^2_7T^*X\oplus\Lambda^2_{21}T^*X.
\end{equation}
The rank-7 bundle $V_7$ defined above can now be identified as \cite{Karigiannis2007Spin7}
\begin{equation}
V_7\cong\Lambda^2_7T^*X.
\label{eq:v7-two-forms}
\end{equation}

This identification describes the target-dependent internal index of the left-moving fermions. The seven components do not arise from a naive decomposition \(TX\cong\mathbf7\oplus\mathbf1\). Instead, they are independent Fermi indices in \(\Lambda^2_7T^*X\). We remark that The fields \(\lambda_-\) remain 2D fermions and are not target-space two-form gauge fields, they are just carrying indices associated with the differential two-form associated to the Cayley 4-form of the Spin(7) structure.

To summarize, both bundles $TX$ and $V_7$ respectively coupling to the right- and left-moving fermions come from only the Spin(7) structure of the target space manifold, with no independent gauge bundle introduced.\footnote{The last condition in \eqref{eq:01-anomaly-condition}, specialized to \(E=V_7\), has a familiar heterotic form. Its de Rham image underlies the Green--Schwarz-type relation \(dH\propto\operatorname{tr}(R\wedge R)-\operatorname{tr}(F_{V_7}\wedge F_{V_7})\), with overall sign and trace normalization fixed by convention. Its integral refinement requires the quantized \(B\)-field cocycle described below. The bundle \(V_7\) plays the algebraic role of the heterotic gauge bundle, but no critical heterotic completion is assumed here.} Therefore, we call this sigma model the canonical Spin(7) sigma model, where ``canonical'' refers to the uniqueness of the left-moving fermion bundle once the Spin(7) target is specified. Note that if one instead takes left-moving fermion bundle to be \(E=TX\), i.e., same as the right-moving one, the sigma model would become the familiar nonchiral \(\mathcal N=(1,1)\) theory for Spin(7) compactification of Type II strings \cite{ShatashviliVafa1994}. The key point of our construction is that it is a genuine chiral theory while still canceling the target-frame anomaly, which we will discuss in the following subsection.

\subsection{Target-frame anomaly cancellation}
\label{subsec:universal-characteristic-identity}

Compared to its non-chiral $\mathcal{N}=(1,1)$ cousin, whose left- and right-moving fermions are both sections of the tangent bundle $TX$, the chiral $\mathcal{N}=(0,1)$ Spin(7) sigma model constructed above looks potentially problematic because the left- and right-moving fermions transform in representations of different ranks. However, the target-frame anomaly actually cancels because the relevant 2D chiral anomaly depends on quadratic charges rather than the dimensions of the representations.

The chiral anomaly of interest is the same charge-squared mechanism as a 2D Abelian chiral gauge anomaly. Up to a common normalization, a \(U(1)\) background whose charges are carried by worldsheet fermions contributes the following anomaly polynomial \cite{GukovPeiPutrov2020}
\begin{equation}
I_4^{U(1)}
=
\left(
\sum_{\psi_+}q^2-
\sum_{\lambda_-}q^2
\right)
\left(\frac{F}{2\pi}\right)^2.
\label{eq:abelian-chiral-anomaly}
\end{equation}
The anomaly vanishes when the charge-squared sums of the two chiralities agree. The Spin(7) calculation below is the non-Abelian version of this equality.

We expose these charges by restricting \(\Spin(7)\) to its maximal torus \(U(1)^3\). This is the largest subgroup of commuting rotations, so every representation decomposes into components with ordinary Abelian charges. Let \(x_1,x_2,x_3\) denote the corresponding curvature coordinates. A weight
\begin{equation}
w=q_1x_1+q_2x_2+q_3x_3
\end{equation}
records the charge vector \((q_1,q_2,q_3)\). For a real representation, the nonzero weights occur in conjugate pairs \(\{w,-w\}\). Such a pair describes one rotating real plane after complexification, so it contributes \(w^2\) once to the first Pontryagin class.\footnote{For an ordinary rotation of a real plane, \(z=u+iv\) and \(\bar z=u-iv\) transform with charges \(+1\) and \(-1\). They are complex conjugates built from the same two real components, so they describe one real plane rather than two independent fermions.} A zero weight is neutral under the chosen maximal torus and contributes nothing to this calculation.\footnote{It need not define a Spin(7) singlet.}

The vector representation \(\rho_7\), used to construct $V_7$ left-moving fermion bundle, has weights \cite{FegerKephart2012}
\begin{equation}
0,\qquad \pm x_1,\qquad\pm x_2,\qquad\pm x_3.
\label{eq:rho7-weights}
\end{equation}
It thus contains three rotating real planes and one real line fixed by the maximal torus. The full representation remains irreducible under Spin(7). Choosing one weight from each conjugate pair gives
\begin{equation}
p_1(\rho_7)=x_1^2+x_2^2+x_3^2.
\label{eq:p1-rho7}
\end{equation}

The representation \(\Delta_8\), associated to the tangent bundle $TX$, is a spin representation of \(\Spin(7)\). A spinor sees half the rotation angle in each plane, so its eight weights are \cite{FegerKephart2012}
\begin{equation}
\frac{1}{2}(\pm x_1\pm x_2\pm x_3),
\label{eq:delta8-weights}
\end{equation}
with all sign choices. They form four conjugate pairs. Taking one representative from each pair gives
\begin{align}
p_1(\Delta_8)
&=\frac{1}{4}\Bigl[
(x_1+x_2+x_3)^2
+(x_1+x_2-x_3)^2
\nonumber\\
&\hspace{27mm}
+(x_1-x_2+x_3)^2
+(-x_1+x_2+x_3)^2
\Bigr]
\nonumber\\
&=x_1^2+x_2^2+x_3^2.
\label{eq:p1-delta8}
\end{align}

\eqref{eq:p1-delta8} therefore gives the same quadratic anomaly coefficient as \eqref{eq:p1-rho7}. Substitution into the chiral anomaly polynomial gives
\begin{equation}
I_{4,\Spin(7)}
\mathrel{\propto}
p_1(\Delta_8)-p_1(\rho_7)
=0.
\label{eq:local-spin7-anomaly-cancellation}
\end{equation}
Thus the Spin(7)-dependent \emph{local} anomaly cancels exactly. 

We now lift this local equality to an integral statement and then check the possible flat remainder. The integral refinement is detected by the large gauge transformations. In the present case, the associated integral cohomology class reads \cite{GrayGreen1970}
\begin{equation}
H^4(B\Spin(7);\ZZ)\cong\ZZ,
\end{equation}
and both \(\rho_7\) and \(\Delta_8\) have Dynkin index one in the same integral normalization \cite{FegerKephart2012}. Since this cohomology group has no torsion, the common index fixes the full degree-four class rather than only its rational image. Let \(u\) denotes the primitive generator, then
\begin{equation}
\lambda(\rho_7)=\lambda(\Delta_8)=u,
\qquad
p_1(\rho_7)=p_1(\Delta_8)=2u.
\label{eq:integral-spin7-generator}
\end{equation}
Furthermore, both representations can be lifted to the corresponding spin groups, so their universal first and second Stiefel--Whitney classes vanish. We thus obtain
\(w_1(\Delta_8)=w_1(\rho_7)=0\) and \(w_2(\Delta_8)=w_2(\rho_7)=0\) for \eqref{eq:01-anomaly-condition}. Together with \eqref{eq:integral-spin7-generator}, these identities promote the maximal-torus calculation to integral target-frame anomaly cancellation.

\paragraph{The $V_7$-twisted String structure and global anomaly cancellation.}
\label{subsec:virtual-string-structure}

The equality of characteristic classes in \eqref{eq:integral-spin7-generator} pulls back along \(X\to B\Spin(7)\) to \(\lambda(TX)-\lambda(V_7)=0\). This vanishing difference guarantees that an anomaly trivialization exists, but it does not select one. Choosing such a trivialization equips a String structure to the following virtual bundle \cite{TachikawaYamashita2021},
\begin{equation}
\Xi_7=TX-V_7,
\label{eq:relative-virtual-bundle}
\end{equation}
where the ``virtual'' means it is not generally an actual line bundle,\footnote{A complex line bundle with vanishing first Chern class admits a nowhere-vanishing section, but the vanishing class does not choose that section. Likewise, \(\lambda(TX)-\lambda(V_7)=0\) guarantees that the nullhomotopy exists but does not choose it.} but just a rank-one class in \(KO(X)\). Intuitively, this virtual bundle records the net chiral fermion bundle. A String structure on \(\Xi_7\) is a chosen nullhomotopy\footnote{Represent \(\lambda(TX)-\lambda(V_7)\) by a map
\(f_\lambda\colon X\to K(\mathbb Z,4)\). A nullhomotopy is a homotopy
from \(f_\lambda\) to the constant map. The vanishing of the
cohomology class guarantees that such a trivialization exists but
does not choose one. Homotopy classes of these choices form an
\(H^3(X;\mathbb Z)\) torsor. This is the homotopy-theoretic meaning
of the trivialization entering a twisted String structure
\cite{Sati:2009ic,TachikawaYamashita2021}.}
\begin{equation}
\alpha\colon
\lambda(\Xi_7)=\lambda(TX)-\lambda(V_7)
\simeq0.
\label{eq:virtual-string-trivialization}
\end{equation}
Equivalently, \(\alpha\) is a \(V_7\)-twisted String structure on \(TX\), whose twist class is \(\lambda(V_7)\). Displaying the difference in \eqref{eq:virtual-string-trivialization} fixes the sign convention for the twist.

At the universal topological level, the homotopy class of \(\alpha\) is unique because \cite{GrayGreen1970}
\begin{equation}
H^3(B\Spin(7);\ZZ)=0.
\end{equation}
We use its pullback to \(X\) as the universal String structure on \(\Xi_7\). Other topological String structures on a fixed \(X\) form an \(H^3(X;\ZZ)\) torsor. A differential refinement additionally requires the connections and a \(B\)-field cocycle. These are choices of the quantum theory beyond the universal bundle identity.

The characteristic classes do not by themselves exclude all possible global anomalies,\footnote{A flat anomaly line bundle can still have holonomy around a loop of backgrounds, just as a flat \(U(1)\) connection can have a nontrivial Wilson line operator.} which might be torsional. To isolate the Spin(7)-dependent part, subtract the pure gravitational rank and consider
\begin{equation}
\Delta_8-\rho_7-\mathbf1
\end{equation}
as a virtual representation of rank and quadratic index zero. Potential torsion anomalies are detected by the spin bordism \cite{Freed:2016rqq}, which in this case reads
\begin{equation}
\Omega_3^{\mathrm{Spin}}(B\Spin(7))=0.
\label{eq:spin7-bordism-vanishing}
\end{equation}
Here the superscript \(\mathrm{Spin}\) specifies the tangential structure of the 3D bulk manifold used for anomaly inflow of the 2D sigma model, while \(B\Spin(7)\) records its background principal \(\Spin(7)\) bundle.\footnote{The role of tangential Spin structure is familiar from the 2D Arf theory, whose partition function depends on the spin structure. The relevant Atiyah--Hirzebruch spectral sequence is \(E^2_{p,q}=H_p(B\Spin(7);\Omega_q^{\mathrm{Spin}})\Rightarrow\Omega_{p+q}^{\mathrm{Spin}}(B\Spin(7))\) \cite{Stong1968}. Since \(B\Spin(7)\) is 3-connected and \(\Omega_3^{\mathrm{Spin}}(\mathrm{pt})=0\), it has no nonzero term of total degree three.} Equation~\eqref{eq:spin7-bordism-vanishing} rules out a residual torsion Spin(7) anomaly after the local and integral classes have cancelled. 

Combining the local \eqref{eq:local-spin7-anomaly-cancellation}, integral \eqref{eq:integral-spin7-generator}, and bordism \eqref{eq:spin7-bordism-vanishing} results shows that the canonical chiral $\mathcal{N}=(0,1)$ sigma model has no Spin(7)-dependent target-frame anomaly. 

\paragraph{The remaining gravitational anomaly.}
\label{subsec:model-remains-chiral}

The target-frame anomaly depends on the weights, but the gravitational anomaly counts net chiral degrees of freedom. The free-field count at large volume gives equal left-moving and right-moving contributions from the eight bosons. The eight right-moving Majorana--Weyl fermions contribute four units to \(c_R\), while the seven left-moving Majorana--Weyl fermions contribute \(7/2\) units to \(c_L\) \cite{GukovPeiPutrov2020}. Hence
\begin{equation}
c_R=12,
\qquad
c_L=\frac{23}{2},
\qquad
c_R-c_L=\frac{1}{2}.
\end{equation}
Thus we are left with a degree-one gravitational anomaly 
\begin{equation}
    \nu=2(c_R-c_L)=1.
\end{equation} The partition function transforms as a section of its gravitational anomaly line bundle rather than as a modular-invariant scalar. Coupling to the inverse 3D invertible spin theory makes the combined system invariant \cite{Freed:1986zx,Freed:2016rqq}. One could also instead add a free left-moving Fermi multiplet to cancel the degree directly in 2D.\footnote{A free left-moving \(\mathcal N=(0,1)\) Fermi multiplet contributes \(c_L=1/2\) and is trivial under the target Spin(7) frame action. It cancels the gravitational degree without changing the target-frame anomaly. Note that this is just a minimal 2D anomaly cancellation but not by itself a critical heterotic string theory completion.}

\section{The Spin(7) charge and its modular refinement}
\label{sec:spin7-smith-charge}

In this section, we determine the degree-one charge left open by the anomaly analysis. Section~\ref{sec:canonical-model} gives \(TX-V_7\) a String structure and fixes the gravitational degree to \(\nu=1\). The degree alone does not decide whether the geometric class constructed below is zero or the nonzero element \(\eta\) in \(\tmf_1\cong\ZZ/2\{\eta\}\). We answer this question by turning the virtual bundle \(TX-V_7\) into the tangent bundle of an actual 1D manifold.

Choose a transverse section of \(V_7\) and let \(Z\) be its zero locus. The normal bundle of \(Z\) is \(V_7|_Z\), which gives
\begin{equation}
(TX-V_7)|_Z\simeq TZ.
\label{eq:section-opening-stable-tangent}
\end{equation}
The String structure on \(TX-V_7\) therefore becomes an ordinary String structure on \(Z\). Its class defines the geometric Spin(7) charge \(\mu_7(X)\). We derive
\begin{equation}
\boxed{
\mu_7(X)=\bigl(\chi(X)\bmod2\bigr)\eta.
}
\label{eq:smith-main}
\end{equation}
The derivation quantizes the seven Fermi zero modes into a rank-eight spinor state bundle and converts the String zero locus into an ordinary twisted Dirac index on \(X\). Spin(7) representation theory identifies this state bundle with \(TX\), and the index theorem reduces its parity to \(\chi(X)\). The same zero locus controls the complete normalized Dirac--Ramond series. Its positive-energy oscillators cancel in the index, and the degree-one cusp isomorphism gives the unique corresponding \(KO_{\mathrm{MF}}\) class.

\subsection{The Fermi zero locus and the odd-rank ABS class}
\label{subsec:string-zero-locus}

Let \(X\) be a closed 8D Spin(7) manifold with a specified Spin(7) structure, and equip
\begin{equation}
TX-V_7
\end{equation}
with the pulled-back universal String structure described in \cref{subsec:virtual-string-structure}. Choose a section
\begin{equation}
s\in\Gamma(X,V_7)
\end{equation}
transverse to the zero section. Its zero locus
\begin{equation}
Z=s^{-1}(0)
\end{equation}
is a closed 1D manifold. Here \(s\) is an auxiliary bosonic section of the Fermi bundle that represents its Euler and Smith classes.\footnote{It is not the worldsheet fermion \(\lambda_-\), and the construction does not require the interacting sigma model to select a physical defect \(Z\).}

Locally, \(s\) is a collection of seven functions. Transversality means that their differentials are independent along \(Z\). More precisely, for every \(x\in Z\), the derivative
\begin{equation}
ds_x\colon T_xX\longrightarrow (V_7)_x
\end{equation}
is surjective and has kernel \(T_xZ\). It gives the exact sequence
\begin{equation}
0\longrightarrow TZ
\longrightarrow TX|_Z
\xrightarrow{\ ds\ }
V_7|_Z
\longrightarrow0.
\label{eq:zero-locus-normal-sequence}
\end{equation}
Thus \(\dim Z=8-7=1\) and
\begin{equation}
N_{Z/X}\cong V_7|_Z.
\end{equation}
This identifies the restriction of the virtual bundle \(TX-V_7\) with the genuine tangent bundle \(TZ\),
\begin{equation}
TZ\simeq (TX-V_7)|_Z.
\label{eq:zero-locus-stable-tangent}
\end{equation}
The String trivialization of \(TX-V_7\) restricts through \eqref{eq:zero-locus-stable-tangent} and gives \(Z\) a String structure.\footnote{The restriction \(TX|_Z\) still has rank eight because only its base points have been restricted. The quotient by \(TZ\) gives the seven normal directions. Also, \(s|_Z=0\) does not trivialize \(V_7|_Z\). This bundle cancels from the stable tangent class because it is the normal bundle of \(Z\).}

Two transverse sections determine the same String bordism class. Let \(s_t\) be a generic transverse homotopy between \(s_0\) and \(s_1\). Its zero locus
\begin{equation}
W=\{(x,t)\in X\times[0,1]\mid s_t(x)=0\}
\end{equation}
has dimension \(9-7=2\) and inherits a String structure by the same stable tangent argument. Its boundary is
\begin{equation}
\partial W=Z_1\sqcup(-Z_0),
\qquad
Z_i=s_i^{-1}(0).
\end{equation}
Thus \(Z_0\) and \(Z_1\) are String bordant even when their components and embeddings differ. We can thus define
\begin{equation}
\mu_7(X)
:=\sigma_{\mathrm{AHR}}[Z]
\in\pi_1\tmf
\cong\ZZ/2\{\eta\},
\label{eq:def-mu7}
\end{equation}
where \(\sigma_{\mathrm{AHR}}\colon M\mathrm{String}\to\tmf\) is the String orientation \cite{AndoHopkinsRezk2010,AndoEtAl2014}.

The group \(\tmf_1=\pi_1(\tmf)\) has two elements. The nonzero element \(\eta\) is represented by the nonbounding circle with periodic spin structure (physically the Ramond condition) and is detected by the mod-two index of one real Majorana zero mode.\footnote{We hope the reader does not confuse this \(\eta\) with the Dedekind eta function used below.} The bounding circle with antiperiodic spin structure represents zero.

\begin{proposition}[Spin(7) String zero locus]
\label{prop:spin7-zero-locus}
The class \(\mu_7(X)\) is independent of the transverse section. It depends on the specified Spin(7) structure and its pulled-back universal String structure on \(TX-V_7\).
\end{proposition}

The zero locus is the geometric representative of the \(KO\) Euler class of \(V_7\).\footnote{The ordinary construction may be familiar. For an oriented rank-\(r\) bundle \(E\to X\), the Euler class \(e(E)\in H^r(X,\mathbb Z)\) is Poincaré dual to the zero locus of a transverse section. When \(E=TX\), one obtains \(\int_Xe(TX)=\chi(X)\). The \(KO\) construction used here refines the same mechanism by retaining the spin and Clifford information carried by the zero locus.} In the index, this Euler class saturates all seven left-moving zero modes. The dimension of \(Z\) records the gravitational degree \(\nu=1\), while its String bordism class determines whether the element of \(\tmf_1\) is \(0\) or \(\eta\).

\paragraph{The odd-rank ABS class.}
\label{subsec:odd-rank-abs}

The symbol \([Z]\) in \eqref{eq:def-mu7} denotes the String bordism class of the 1D zero locus. Its image under \(\sigma_{\mathrm{AHR}}\) is \(\mu_7(X)\), so identifying this class decides whether the \(\tmf\) charge is \(0\) or \(\eta\). To compute \([Z]\), we use the real \(KO\) Euler class of a rank-seven spin bundle.

Physically speaking, this amounts to quantizing the seven left-moving fermion zero modes. At a fixed point \(x\in X\), these zero modes are operators that generate the Clifford algebra \(Cl(E_x)\). They act on the ungraded real rank-eight spinor module \(S(E)_x\). The odd number of generators gives no intrinsic fermion-parity grading on this module. As \(x\) varies, the modules form a rank-eight state bundle
\begin{equation}
S(E)\longrightarrow X.
\end{equation}
The zero modes are the Clifford operators, while \(S(E)\) is the finite-dimensional state space on which they act.\footnote{Its rank is the dimension of this state space and does not introduce an eighth Fermi multiplet.} It is also not the full sigma-model Hilbert space, which includes the nonzero oscillators considered below.

The Atiyah--Bott--Shapiro construction \cite{AtiyahBottShapiro1964} represents the zero-locus charge by an operator family that is invertible away from the zero section. Let \(E\to X\) be a real spin bundle of rank seven. Odd-dimensional spinors do not split into \(S^+\) and \(S^-\), so we first add an auxiliary trivial line and set
\begin{equation}
W=E\oplus\RR.
\end{equation}
This stabilization supplies the chiral spinor grading required by the even-rank ABS construction. It does not add a physical target direction or a new Fermi multiplet. The Atiyah--Bott--Shapiro Thom class of the rank-eight bundle \(W\) is represented by Clifford multiplication
\begin{equation}
c(v)\colon\pi^*S_W^+\longrightarrow\pi^*S_W^-
\label{eq:rank8-abs-symbol}
\end{equation}
over the total space of \(W\) \cite{Atiyah1968}. With the positive Clifford convention of Appendix~\ref{app:conventions}, the full operator
\begin{equation}
T(v)=\sum_{a=1}^8v^a\Gamma_a
\end{equation}
obeys
\begin{equation}
\begin{aligned}
T(v)^2
&=
\frac12\sum_{a,b}v^av^b\{\Gamma_a,\Gamma_b\}
=|v|^2\operatorname{id},\\
T(v)^{-1}
&=
\frac{T(v)}{|v|^2},
\qquad v\neq0.
\end{aligned}
\label{eq:abs-clifford-invertibility}
\end{equation}
Thus the ABS operator is invertible away from \(v=0\), and its \(KO\) class is supported on the zero section.\footnote{The physical model for this construction is Witten's realization of a Type I D-string as a tachyon soliton on eight D9-branes and eight anti-D9-branes \cite{Witten:1998cd}. The two real chiral spinor bundles \(S_W^+\) and \(S_W^-\) supply the Chan--Paton bundles, while \(c(v)=v^a\Gamma_a\) is the open-string tachyon profile in the eight directions transverse to the string. Away from \(v=0\), the tachyon is invertible and the branes and antibranes condense to the vacuum. At \(v=0\), the tachyon vanishes and leaves a D-string localized on the zero locus. The compactly supported \(KO\) class of this tachyon configuration is its conserved Type I D-brane charge.}

We now restrict \eqref{eq:rank8-abs-symbol} to the auxiliary real line with coordinate \(t\). Under \(\Spin(7)\subset\Spin(8)\), both half-spinor bundles restrict to the real rank-eight spinor bundle \(S(E)\), and the symbol becomes
\begin{equation}
t\,\mathrm{id}_{S(E)}.
\label{eq:abs-auxiliary-line}
\end{equation}
The relevant Clifford algebras are related by a graded Morita equivalence \cite[Appendix~B]{DebrayEtAl2026Smith},
\begin{equation}
Cl_8\widehat\otimes Cl_{-1}
\sim_{\mathrm{Morita}}
Cl_7.
\label{eq:graded-morita}
\end{equation}
This is an equivalence of graded module categories, not an isomorphism of
algebras. It converts the stabilized symbol into the odd-rank Euler class.

\begin{proposition}[Odd-rank ABS formula]
\label{prop:odd-rank-abs}
For a rank-seven spin bundle \(E\to X\), the Euler class \(e_{KO}(E)\in KO^7(X)\) is
\begin{equation}
e_{KO}(E)
=\eta\,\beta^{-1}\bigl([S(E)]-8\bigr),
\label{eq:odd-rank-ko-euler}
\end{equation}
with the Clifford degree understood through \eqref{eq:graded-morita}. If a transverse section has zero locus \(Z\subset X^8\), then
\begin{equation}
\operatorname{ABS}[Z]
=\Bigl(\ind_{\mathbb C}\bigl(D_X\otimes S(E)_{\mathbb C}\bigr)\bmod2\Bigr)\eta
\in KO_1.
\label{eq:abs-twisted-index}
\end{equation}
\end{proposition}

The subtraction by \(8\) in \eqref{eq:odd-rank-ko-euler} removes the trivial rank of the zero-mode state bundle and retains its twisting over \(X\). If \(E\) is trivial, then \(S(E)\cong\RR^8\) and \([S(E)]-8=0\), as required because a trivial bundle has a nowhere-zero section. The factor \(\beta^{-1}\) implements the Bott degree shift, while \(\eta\) records the remaining odd Clifford degree.

\begin{proof}
Let \(u_W\in KO_c^8(W)\) be the rank-eight ABS Thom class. The Thom-space
identification
\begin{equation}
\operatorname{Th}(E\oplus\RR)
\cong
\Sigma\operatorname{Th}(E)
\end{equation}
allows the rank-eight class \(u_W\) to determine the rank-seven Thom class \cite[Theorem~6.1]{Atiyah1968}. On the suspension
coordinate, the ABS symbol is \(t\,\mathrm{id}_{S(E)}\). As spelled out in
Appendix~\ref{app:suspension-eta}, in compact support \(KO\) the 1D symbol
\(\delta_t\) contributes \(\eta\), while the Bott regrading contributes \(\beta^{-1}\),
and the remaining twisting bundle is \(S(E)\). Applying the graded Morita functor therefore gives
\begin{equation}
e_{KO}(E)=\eta\,\beta^{-1}[S(E)].
\end{equation}
Writing the answer as a reduced class gives
\eqref{eq:odd-rank-ko-euler}, since \(8\eta=0\). In particular, the trivial
rank-seven bundle has zero Euler class, as required by its nowhere-zero
section.

The Euler--Gysin identity identifies the pushforward of this class with the
pushforward from the transverse zero locus. For a closed 8D spin manifold,
\begin{align}
p^{KO}_!e_{KO}(E)
&=
\eta\left(
\ind_{\mathbb C}(D_X\otimes S(E)_{\mathbb C})
-8\widehat A(X)
\right)
\nonumber\\
&=
\left(
\ind_{\mathbb C}(D_X\otimes S(E)_{\mathbb C})\bmod2
\right)\eta.
\end{align}
Here complexification in degree eight has multiplicity one, and the second
term vanishes because \(8\eta=0\). The left-hand side is
\(p^{KO}_{Z!}(1)=\operatorname{ABS}[Z]\), which proves
\eqref{eq:abs-twisted-index}.
\end{proof}

\subsection{The Spin(7) twisted index and Euler parity}
\label{subsec:spin7-twisted-index}

The Spin(7) representation now turns the zero-mode state bundle into a familiar geometric bundle. For \(E=V_7\), the real spinor representation of \(\rho_7\) is \(\Delta_8\). Both \(S(V_7)\) and \(TX\) are therefore associated to the same principal Spin(7) bundle through the same representation,
\begin{equation}
S(V_7)
=P\times_{\Spin(7)}S(\rho_7)
\cong
P\times_{\Spin(7)}\Delta_8
=TX.
\label{eq:spinor-v7-associated}
\end{equation}
Thus
\begin{equation}
S(V_7)\cong TX
\label{eq:spinor-v7-tangent}
\end{equation}
as real Spin(7) bundles. The two bundles have the same Spin(7) action and the same transition functions. Physically, quantizing the seven left-moving Fermi zero modes produces the tangent-representation state bundle.

Proposition~\ref{prop:odd-rank-abs} now reduces the Smith charge to the parity of one integer,
\begin{equation}
I_{TX}:=\ind_{\mathbb C}\bigl(D_X\otimes TX_{\mathbb C}\bigr).
\label{eq:def-tangent-twisted-index}
\end{equation}
The Atiyah--Singer theorem \cite{Atiyah:1968mp} gives
\begin{align}
I_{TX}
&=\left\langle \widehat A(TX)\operatorname{ch}(TX_{\mathbb C}),[X]\right\rangle
\nonumber\\
&=\left\langle
\frac{37p_1^2-124p_2}{720},[X]
\right\rangle.
\label{eq:tangent-twisted-index-characteristic}
\end{align}

As reviewed in \cite{Joyce2002Exceptional}, the Euler and Pontryagin classes of a rank-eight bundle with Spin(7) structure obey
\begin{equation}
4p_2-p_1^2=8e.
\label{eq:spin7-characteristic-relation}
\end{equation}
Combining this with
\begin{equation}
\widehat A(X)
=\left\langle\frac{7p_1^2-4p_2}{5760},[X]\right\rangle
\end{equation}
and the Hirzebruch signature formula gives two useful forms of the same index,
\begin{equation}
I_{TX}
=8\widehat A(X)-\frac{\chi(X)}{3}
=24\widehat A(X)-\Sign(X).
\label{eq:index-ahat-euler-signature}
\end{equation}
The second expression immediately determines its parity. Write the middle Betti number as
\begin{equation}
b_4=b_4^++b_4^-,
\qquad
\Sign(X)=b_4^+-b_4^-.
\end{equation}
It follows that \(\Sign(X)\equiv b_4(X)\pmod2\). Poincaré duality on a closed oriented 8D manifold gives
\begin{equation}
\chi(X)
=2\bigl(b_0-b_1+b_2-b_3\bigr)+b_4,
\end{equation}
and hence
\begin{equation}
\Sign(X)\equiv b_4(X)\equiv\chi(X)\pmod2.
\end{equation}
Since \(24\widehat A(X)\) is even,
\begin{equation}
I_{TX}\equiv\chi(X)\pmod2.
\label{eq:index-euler-parity}
\end{equation}

We can now evaluate the String class itself. In degree one, the maps
\begin{equation}
\Omega_1^{\mathrm{String}}
\longrightarrow
\Omega_1^{\mathrm{Spin}}
\longrightarrow
KO_1
\label{eq:degree1-bordism-ko}
\end{equation}
are isomorphisms, and the String orientation sends the nonbounding circle to \(\eta\in\tmf_1\).

\begin{theorem}[Spin(7) Smith formula]
\label{thm:spin7-smith-formula}
For every closed 8D Spin(7) manifold with a specified Spin(7) structure and the pulled-back universal String structure on \(TX-V_7\),
\begin{equation}
\boxed{
\mu_7(X)
=\bigl(\chi(X)\bmod2\bigr)\eta
\in\tmf_1.
}
\end{equation}
\end{theorem}

\begin{proof}
The zero locus defines \(\mu_7(X)\) by \eqref{eq:def-mu7}. Under the isomorphisms \eqref{eq:degree1-bordism-ko}, its image is the ABS class in \eqref{eq:abs-twisted-index}. Equations~\eqref{eq:spinor-v7-tangent} and \eqref{eq:index-euler-parity} identify this class with \((\chi(X)\bmod2)\eta\). Since \(\Omega_1^{\mathrm{String}}\to KO_1\) is an isomorphism, the \(KO\) computation determines the String bordism class and its image in \(\tmf_1\).
\end{proof}

We summarize the geometric and index-theoretic dictionary in the following schematic chain,
\begin{equation}
\boxed{
\begin{aligned}
(TX-V_7)|_Z\simeq TZ
&\longrightarrow
[Z]_{\mathrm{String}}
\longrightarrow
[Z]_{\mathrm{ABS}}
\\
&\longrightarrow
\Bigl(\ind_{\mathbb C}\!\bigl(D_X\otimes S(V_7)_{\mathbb C}\bigr)\bmod2\Bigr)\eta,
\\
&\xrightarrow{\ S(V_7)\cong TX\ }
\Bigl(\ind_{\mathbb C}\!\bigl(D_X\otimes TX_{\mathbb C}\bigr)\bmod2\Bigr)\eta
\\
&=
\bigl(\chi(X)\bmod2\bigr)\eta.
\end{aligned}
}
\label{eq:spin7-smith-dictionary}
\end{equation}
The first two arrows convert the virtual rank-one bundle into actual 1D String geometry and then apply the ABS orientation. The Spin(7) representation identifies the twisting bundle with \(TX\), and the index theorem reduces its parity to the Euler parity.

\subsection{Dirac--Ramond localization and oscillator cancellation}
\label{subsec:dirac-ramond-localization}

The Smith formula determines the zero-mode contribution. We now ask whether the positive-energy oscillators carry additional degree-one information. After removing the universal vacuum contribution, the entire real-\(KO\) Dirac--Ramond series localizes to the same 1D String locus. For the canonical Spin(7) model,
\begin{equation}
\boxed{
\Phi_2(X,V_7;q)
=
\bigl(\chi(X)\bmod 2\bigr)\eta
}
\label{eq:spin7-full-ko-series}
\end{equation}
as an equality of complete normalized \(q\)-series. In particular, every positive \(q\)-coefficient vanishes.

The mechanism is the oscillator version of the zero-locus construction. The \(KO\) Euler class saturates the left-moving fermion zero modes and restricts the index to \(Z\). Along this locus, the normal bosonic oscillators and the \(E\)-valued fermionic oscillators cancel. Only the tangent oscillators of \(Z\) remain. For a 1D zero locus, their normalized virtual bundle is trivial.

Let \(X^d\) be a closed spin manifold and let \(E\to X\) be a real spin bundle of rank \(r\). We write
\begin{equation}
\nu=d-r.
\end{equation}
The real-\(KO\) Euler class of \(E\) is denoted by \(e_{KO}(E)\). We define
\begin{equation}
\begin{split}
\Phi_2(X,E;q)
:=
p^{KO}_!\Bigg[
e_{KO}(E)
\prod_{m\geq 1}
\Sym_{q^m}
\bigl(TX-\RR^d\bigr)
\Lambda_{-q^m}
\bigl(E-\RR^r\bigr)
\Bigg]
\\
\in KO^{-\nu}(\mathrm{pt})[[q]],
\end{split}
\label{eq:relative-ko-dirac-ramond}
\end{equation}
where \(p\colon X\to\mathrm{pt}\) is the spin pushforward.

Each factor has a direct Fock-space meaning. The bosonic oscillators produce \(\Sym_{q^m}\), while the left-moving fermionic supertrace produces \(\Lambda_{-q^m}\). The minus sign records \((-1)^F\) and does not represent a negative Hilbert space.\footnote{For a one-particle state space \(V\), the bosonic and fermionic Fock spaces are encoded by \(\Sym_t(V)=\sum_{n\geq0}t^n\Sym^n(V)\) and \(\Lambda_t(V)=\sum_{n\geq0}t^n\Lambda^n(V)\). A bosonic oscillator of energy \(m\) gives \(1+q^m+q^{2m}+\cdots=(1-q^m)^{-1}\), while a fermionic oscillator contributes \(1-q^m\) to the supertrace.} The subtraction of the trivial bundles removes the universal vacuum contribution and produces the normalized index \cite{GukovPeiPutrov2020,TachikawaYamashitaYonekura2023}.

Using
\begin{equation}
\Sym_t(W)=\Lambda_{-t}(W)^{-1},
\end{equation}
we can rewrite the oscillator contribution as
\begin{equation}
\prod_{m\geq 1}
\Sym_{q^m}
\bigl(TX-E-\RR^\nu\bigr).
\label{eq:relative-oscillator-bundle}
\end{equation}
It is useful to introduce
\begin{equation}
\Theta_q(W):=
\prod_{m\geq 1}\Sym_{q^m}(W).
\end{equation}
Then
\begin{equation}
\Phi_2(X,E;q)
=
p^{KO}_!
\left[
e_{KO}(E)
\Theta_q\bigl(TX-E-\RR^\nu\bigr)
\right].
\label{eq:relative-ko-series-compact}
\end{equation}
The use of real \(KO\), rather than complex \(K\)-theory, is essential. The degree-one class is two-torsion and disappears after passing to rational characteristic forms.

\paragraph{Localization at the Fermi zero locus.}
\label{subsec:ko-smith-localization}

Choose a section \(s\in\Gamma(X,E)\) transverse to the zero section and let
\begin{equation}
i\colon Z=s^{-1}(0)\hookrightarrow X.
\end{equation}
The zero locus is a closed spin manifold of dimension \(\nu\), and its stable tangent bundle satisfies
\begin{equation}
TZ\simeq i^*(TX-E).
\label{eq:stable-tangent-zero-locus}
\end{equation}

\begin{theorem}[Dirac--Ramond--Smith localization]
\label{thm:dirac-ramond-smith}
For every closed spin pair \((X^d,E^r)\) and every transverse section of \(E\),
\begin{equation}
\boxed{
\Phi_2(X,E;q)
=
p^{KO}_{Z!}
\prod_{m\geq 1}
\Sym_{q^m}
\bigl(TZ-\RR^\nu\bigr).
}
\label{eq:ko-smith-localization}
\end{equation}
Thus the normalized Dirac--Ramond series of the pair \((X,E)\) equals the normalized real-\(KO\) Witten series of its zero locus.
\end{theorem}

\begin{proof}
The spin Euler class is the Gysin image of the unit \cite[Section~4.1]{DebrayEtAl2026Smith},
\begin{equation}
e_{KO}(E)=i^{KO}_!(1).
\label{eq:ko-euler-gysin}
\end{equation}
The projection formula gives
\begin{align}
\Phi_2(X,E;q)
&=
p^{KO}_!
\left[
i^{KO}_!(1)
\Theta_q\bigl(TX-E-\RR^\nu\bigr)
\right]
\nonumber\\
&=
p^{KO}_{Z!}
i^*\Theta_q\bigl(TX-E-\RR^\nu\bigr)
\nonumber\\
&=
p^{KO}_{Z!}
\Theta_q\bigl(TZ-\RR^\nu\bigr),
\end{align}
where the last equality follows from \eqref{eq:stable-tangent-zero-locus}. The argument applies coefficient by coefficient and therefore holds in the \(q\)-adic completion.
\end{proof}

This proof makes the cancellation transparent. The \(KO\) Euler class places the index on \(Z\), while the normal bundle \(E|_Z\) disappears from the oscillator class. This is the index-theoretic form of the cancellation between the normal bosonic modes and the left-moving Fermi modes.

When \(\nu=1\), every oriented 1D manifold has
\begin{equation}
TZ\cong\RR
\end{equation}
as an oriented real bundle. Its spin structure is still detected by the pushforward. The oscillator class itself becomes
\begin{equation}
\Theta_q(TZ-\RR)=1.
\end{equation}
Theorem~\ref{thm:dirac-ramond-smith} therefore gives
\begin{equation}
\boxed{
\Phi_2(X,E;q)
=
p^{KO}_{Z!}(1)
=
\operatorname{ABS}[Z].
}
\label{eq:degree-one-ko-localization}
\end{equation}
This is an equality of the entire normalized series, not only of its constant term.

For the canonical Spin(7) model,
\begin{equation}
d=8,
\qquad
E=V_7,
\qquad
r=7,
\qquad
\nu=1.
\end{equation}
The Smith calculation above identified the ABS class of the zero locus.
\begin{equation}
\operatorname{ABS}[Z]
=
\bigl(\chi(X)\bmod2\bigr)\eta.
\end{equation}
Equation~\eqref{eq:degree-one-ko-localization} now proves the following result.

\begin{corollary}[The complete normalized Spin(7) series]
\label{cor:spin7-complete-ko-series}
For every closed 8D Spin(7) manifold,
\begin{equation}
\Phi_2(X,V_7;q)
=
\bigl(\chi(X)\bmod2\bigr)\eta
\in KO^{-1}(\mathrm{pt})[[q]].
\end{equation}
All positive powers of \(q\) vanish.
\end{corollary}

The vanishing of the positive \(q\)-coefficients does not mean that the sigma model has no excited states. It means that their contributions cancel in this supersymmetric index. The spin structure of \(Z\) remains in the pushforward \(p^{KO}_{Z!}(1)\), so the constant series can still equal the nonzero class \(\eta\).

\paragraph{Normalized and unnormalized indices.}
The class \(\Phi_2(X,E;q)\) is the complete normalized real-\(KO\) BPS index series. When the QFT index is identified with the sigma-model expression, the unnormalized mod-two elliptic genus of a QFT with gravitational degree \(\nu=1\) takes the form
\begin{equation}
I_2(q)
=
\eta_{\mathrm{Ded}}(q)^{-1}
\Phi_2(X,E;q).
\label{eq:raw-mod-two-genus}
\end{equation}
Even when \(\Phi_2\) is constant, the Dedekind factor carries nontrivial \(q\)-dependence. The localization therefore fixes the complete normalized BPS index.\footnote{The full torus partition function \(\cZ_{\mathrm{CFT}}(\tau,\bar\tau)\), the raw mod-two elliptic genus \(I_2(q)\), and the normalized series \(\Phi_2(q)\) are three different observables. The full partition function contains both BPS and non-BPS states and is a section of the gravitational anomaly line. The raw genus is a BPS index. We compute \(\Phi_2(q)\), not the full torus partition function.}

\subsection{The degree-one \texorpdfstring{\(KO_{\mathrm{MF}}\)}{KO(MF)} class}
\label{subsec:spin7-ko-series}
\label{subsec:komf-lift}

The localized series is a real-\(KO\) \(q\)-expansion near the cusp \(q=0\). A cusp expansion is local information on the moduli of elliptic curves and need not determine a global modular-topological class. The spectrum \(KO_{\mathrm{MF}}\) records real-\(KO\) cusp data compatible with modular-form information. It is defined as the homotopy pullback of commutative ring spectra \cite{BerwickEvans2023}
\begin{equation}
KO_{\mathrm{MF}}
=
H_{\mathrm{MF}}
\times_{H_{\mathbb C((q))[u^{\pm2}]}}
KO((q)),
\label{eq:komf-pullback}
\end{equation}
where \(\mathrm{MF}\) denotes the graded ring of weakly holomorphic modular forms with weight-\(k\) forms in degree \(2k\), the coefficients \(\mathbb C((q))[u^{\pm2}]\) with \(|u^2|=-4\) sit in degrees divisible by four, the map from \(KO((q))\) is the Pontryagin character, and the map from \(H_{\mathrm{MF}}\) is \(q\)-expansion. Through the universal property of \eqref{eq:komf-pullback}, the Miller character and the Chern--Dold character assemble into
\begin{equation}
\TMF\longrightarrow KO_{\mathrm{MF}}\longrightarrow KO((q)).
\label{eq:komf-maps}
\end{equation}
The question is whether the class computed above lifts through the second map, the cusp projection.

\begin{lemma}[Degree-one cusp isomorphism]
\label{lem:pi1-komf}
The cusp projection induces an isomorphism
\begin{equation}
q\text{-exp}\colon
\pi_1KO_{\mathrm{MF}}
\xrightarrow{\ \simeq\ }
\pi_1KO((q))
\cong\ZZ/2((q)).
\label{eq:degree-one-cusp-isomorphism}
\end{equation}
\end{lemma}

\begin{proof}
The Mayer--Vietoris sequence of the homotopy pullback \eqref{eq:komf-pullback} reads
\begin{equation}
\pi_2H_{\mathbb C((q))[u^{\pm2}]}
\longrightarrow
\pi_1KO_{\mathrm{MF}}
\longrightarrow
\pi_1KO((q))\oplus\pi_1H_{\mathrm{MF}}
\longrightarrow
\pi_1H_{\mathbb C((q))[u^{\pm2}]}.
\end{equation}
The outer groups vanish because \(\mathbb C((q))[u^{\pm2}]\) is concentrated in degrees divisible by four. The summand \(\pi_1H_{\mathrm{MF}}\) vanishes because \(\mathrm{MF}\) is concentrated in even degrees. The remaining map is therefore an isomorphism onto \(\pi_1KO((q))=KO_1((q))\cong\ZZ/2((q))\).
\end{proof}

Surjectivity gives existence of the lift, while injectivity gives uniqueness.

Let \((X^d,E^r)\) be a closed spin pair and let
\begin{equation}
\alpha\colon\lambda(TX)-\lambda(E)\simeq0
\end{equation}
be an integral String trivialization of the virtual bundle \(TX-E\). A transverse zero locus \(Z=s^{-1}(0)\) inherits a String structure \(\alpha_Z\) through \eqref{eq:stable-tangent-zero-locus}. We define
\begin{equation}
\mu_E^{KO_{\mathrm{MF}}}(X;\alpha)
:=
\left(
M\mathrm{String}
\xrightarrow{\sigma_{\mathrm{AHR}}}
\TMF
\longrightarrow
KO_{\mathrm{MF}}
\right)
[Z,\alpha_Z]
\in\pi_\nu KO_{\mathrm{MF}}.
\label{eq:komf-smith-class}
\end{equation}
A transverse homotopy of sections gives a String bordism between the corresponding zero loci. Hence \eqref{eq:komf-smith-class} is independent of the chosen section.

The \(KO_{\mathrm{MF}}\) index theorem identifies the String zero-locus class with its analytic index \cite{BerwickEvans2023}. Its cusp image is the normalized real-\(KO\) Dirac--Ramond series of \(Z\). We write this series as
\begin{equation}
\Phi_2(Z;q)
:=
p^{KO}_{Z!}\Theta_q\bigl(TZ-\RR^\nu\bigr).
\end{equation}
Combining this comparison with Theorem~\ref{thm:dirac-ramond-smith} gives
\begin{equation}
\boxed{
q\text{-exp}
\left(
\mu_E^{KO_{\mathrm{MF}}}(X;\alpha)
\right)
=
\Phi_2(Z;q)
=
\Phi_2(X,E;q).
}
\label{eq:komf-cusp-localization}
\end{equation}
Thus the same Fermi zero locus determines both the complete cusp series and its \(KO_{\mathrm{MF}}\) lift. For this zero-locus class, separate \(KO_{\mathrm{MF}}\) orientations of \(X\) and \(E\) are not required.\footnote{One may choose a geometric refinement of \(\alpha_Z\) to represent the analytic index. The resulting topological class is independent of that choice.}

\begin{remark}
\label{rmk:degree-one-elementary}
In degree one, the comparison \eqref{eq:komf-cusp-localization} is elementary. The group \(\Omega_1^{\mathrm{String}}\cong\ZZ/2\) is generated by the nonbounding circle, whose class is the image of the stable Hopf element under the unit map from the sphere spectrum to \(M\mathrm{String}\). Every map in \(M\mathrm{String}\to\TMF\to KO_{\mathrm{MF}}\to KO((q))\) is a map of ring spectra and therefore carries \(\eta\) to \(\eta\). The cusp image of \(\mu_E^{KO_{\mathrm{MF}}}(X;\alpha)\) is then the constant series \(\operatorname{ABS}[Z]\in KO_1\subset KO_1((q))\), in agreement with \eqref{eq:degree-one-ko-localization}. The degree-one results below therefore use the analytic index theorem of \cite{BerwickEvans2023} only for context, not as an input.
\end{remark}

\begin{theorem}[Degree-one \texorpdfstring{\(KO_{\mathrm{MF}}\)}{KO(MF)} lift]
\label{thm:degree-one-komf}
Let \((X^d,E^r,\alpha)\) be a closed spin pair with an integral String trivialization of \(TX-E\) and \(d-r=1\). Then \(\Phi_2(X,E;q)\) determines the complete \(KO_{\mathrm{MF}}\) class \(\mu_E^{KO_{\mathrm{MF}}}(X;\alpha)\) through the isomorphism \eqref{eq:degree-one-cusp-isomorphism}.

For \(E=V_7\), let
\begin{equation}
\widetilde\eta\in\pi_1KO_{\mathrm{MF}}
\end{equation}
be the unique element satisfying
\begin{equation}
q\text{-exp}(\widetilde\eta)=\eta.
\end{equation}
Then
\begin{equation}
\boxed{
\mu_{V_7}^{KO_{\mathrm{MF}}}(X)
=
\bigl(\chi(X)\bmod2\bigr)\widetilde\eta.
}
\label{eq:spin7-komf-class}
\end{equation}
\end{theorem}

\begin{proof}
Lemma~\ref{lem:pi1-komf} gives the isomorphism \eqref{eq:degree-one-cusp-isomorphism}. Equation~\eqref{eq:komf-cusp-localization} and Corollary~\ref{cor:spin7-complete-ko-series} give
\begin{equation}
q\text{-exp}
\left(
\mu_{V_7}^{KO_{\mathrm{MF}}}(X)
\right)
=
\bigl(\chi(X)\bmod2\bigr)\eta.
\end{equation}
Injectivity of the cusp projection gives \eqref{eq:spin7-komf-class}.
\end{proof}

The theorem is particularly simple because the degree-one class is torsion. The modular-form component carries no odd-degree information, while the real-\(KO\) cusp detects the complete class. For the Fermat quotient,
\begin{equation}
\mu_{V_7}^{KO_{\mathrm{MF}}}(Y_{\mathrm F})
=
\widetilde\eta.
\end{equation}

The cusp image of the String zero-locus class \(\mu_E^{KO_{\mathrm{MF}}}(X;\alpha)\) is \(\Phi_2(X,E;q)\), and the degree-one cusp isomorphism recovers this class uniquely. This determines the class-level \(KO_{\mathrm{MF}}\) lift. Constructing a pair-level \(KO_{\mathrm{MF}}\) orientation, a local \(2|1\) field-theory representative, and its identification with the interacting sigma model remain open problems, which will be discussed in \cref{sec:discussion}.

\section{Joyce orbifolds}
\label{sec:joyce-orbifolds}

We now place the chiral Spin(7) sigma model at the standard Joyce orbifold point. The corresponding smooth Joyce target realizes the trivial element \(0\in\tmf_1\). We separate this geometric result from the finite-spin anomaly that determines whether the Joyce symmetry can be gauged to define the singular orbifold CFT.

\subsection{The chiral theory at the Joyce point}
\label{subsec:joyce-canonical-orbifold}

Consider the standard Joyce action \(\Gamma=\langle\alpha,\beta,\gamma,\delta\rangle\cong(\ZZ_2)^4\) on \(T^8\) \cite{Joyce1996Spin7,Sugiyama:2001qh}.
\begin{align}
\alpha(x)&=(-x_1,-x_2,-x_3,-x_4,x_5,x_6,x_7,x_8),\nonumber\\
\beta(x)&=(x_1,x_2,x_3,x_4,-x_5,-x_6,-x_7,-x_8),\nonumber\\
\gamma(x)&=\left(\tfrac12-x_1,\tfrac12-x_2,x_3,x_4,
\tfrac12-x_5,\tfrac12-x_6,x_7,x_8\right),\nonumber\\
\delta(x)&=\left(-x_1,x_2,\tfrac12-x_3,x_4,
-x_5,x_6,\tfrac12-x_7,x_8\right).
\label{eq:joyce-action}
\end{align}
Each generator is an isometry of the flat torus whose linear action preserves the Cayley form. It therefore defines a finite internal global symmetry of the \(T^8\) worldsheet theory and preserves the right-moving supercurrent.

The eight bosons and eight right-moving Majorana fermions transform in the restricted \(\Delta_8\) representation. The seven left-moving Majorana fermions transform in \(\rho_7\). The chiral orbifold has field content
\begin{equation}
\mathcal C^{(0,1)}_{\mathrm{Joyce},7}
=
\left[
T^8\text{ bosons}
\oplus8\psi_R^{\Delta_8}
\oplus7\lambda_L^{\rho_7}
\right]\big/\Gamma.
\label{eq:joyce-canonical-orbifold}
\end{equation}
Orbifolding means gauging \(\Gamma\). The finite internal anomaly must therefore vanish, but this requirement is independent of the gravitational anomaly caused by the unequal numbers of left-moving and right-moving fermions. We find
\begin{equation}
\mathcal A^{\mathrm{int}}_\Gamma
(\Delta_8-\rho_7-\mathbf1)=0,
\qquad
c_R-c_L=\frac12.
\label{eq:joyce-physics-summary}
\end{equation}
The first result allows the finite gauging after a choice of anomaly trivialization. The second result says that the gauged theory remains gravitationally anomalous. Its partition function is a section of the gravitational anomaly line rather than a modular-invariant function. We derive the finite-spin result in \cref{subsec:joyce-finite-spin}.

The target geometry answers a different question. Its unweighted orbifold Euler characteristic follows from the commuting-pair formula \cite{BryanFulman1998}.
\begin{equation}
\chi_{\mathrm{orb}}(T^8,\Gamma)
=
\frac1{16}\sum_{g,h\in\Gamma}\chi\bigl((T^8)^{g,h}\bigr).
\label{eq:joyce-standard-commuting-pairs}
\end{equation}
Let \(g_0=\alpha\beta\), which reflects all eight coordinates. The line \(\langle g_0\rangle\) gives the three ordered pairs \((1,g_0)\), \((g_0,1)\), and \((g_0,g_0)\). The half-shifts leave one compatible rank-two sign plane, \(\langle\alpha,\beta\rangle\), which gives six ordered generating pairs. Each of these nine pairs has \(2^8=256\) isolated common fixed points. Hence
\begin{equation}
\boxed{
\chi_{\mathrm{orb}}(T^8,\Gamma)
=
\frac{(3+6)2^8}{2^4}
=144.
}
\label{eq:joyce-standard-euler}
\end{equation}
This direct orbifold calculation gives \(144\) and reproduces the known Euler characteristic of the first smooth Joyce example \cite{Joyce1996Spin7,Sugiyama:2001qh}. Let \(X_{\mathrm J}\) denote that smooth Spin(7) resolution. Applying \cref{thm:spin7-smith-formula}, we obtain
\begin{equation}
\boxed{
\begin{gathered}
\chi(X_{\mathrm J})=144,\\
\mu_7(X_{\mathrm J})
=
(144\bmod2)\eta
=
0
\in\tmf_1\cong\ZZ/2\{\eta\}.
\end{gathered}
}
\label{eq:joyce-tmf-class}
\end{equation}
The smooth Joyce target therefore realizes the trivial element rather than the nonzero generator \(\eta\). This geometric class lies in degree one and does not remove the gravitational anomaly in \eqref{eq:joyce-physics-summary}. Equation \eqref{eq:joyce-standard-commuting-pairs} is a geometric fixed-point calculation.\footnote{For the ordinary nonchiral \(\mathcal N=(1,1)\) orbifold, the same commuting-pair sum is the Ramond--Ramond Witten index. One group element labels the spatially twisted sector and the other is inserted along Euclidean time to implement the projection. This interpretation does not turn the geometric sum into a protected index of the standalone chiral theory.}

The same Joyce action appears in two familiar worldsheet constructions, but neither one is the chiral CFT in \eqref{eq:joyce-canonical-orbifold}. In the heterotic construction of Sugiyama and Yamaguchi, sixteen pre-GSO fermions split into seven and nine, and the gauge sector uses the conformal embedding
\begin{equation}
E_{8,1}\supset\mathfrak{so}(7)_1\oplus\mathfrak{so}(9)_1.
\label{eq:joyce-heterotic-conformal-embedding}
\end{equation}
The GSO projection and the extension to \(E_{8,1}\) correlate the vacuum, vector, and spinor sectors of the two current algebras. The physical Hilbert space no longer factorizes into seven-fermion and nine-fermion Hilbert spaces, so one cannot divide the modular-invariant heterotic partition function by an \(\mathfrak{so}(9)_1\) character. The ordinary \(\mathcal N=(1,1)\) Joyce orbifold is also different. It has tangent fermions in both chiralities, \(c_L=c_R=12\), and no gravitational anomaly. We give the full character identities and protected genera of these comparison theories in \cref{app:joyce-worldsheet-comparisons}.

The standard Joyce point therefore realizes a consistent finite gauging with a nonzero gravitational anomaly, while its smooth target has the trivial geometric class in \eqref{eq:joyce-tmf-class}. We next derive the gauging statement, and then ask whether varying the Cayley-preserving signs and half-shifts can produce an odd Euler characteristic.

\subsection{The finite-spin anomaly}
\label{subsec:joyce-finite-spin}

We now determine whether the finite symmetry in \eqref{eq:joyce-canonical-orbifold} can be gauged. The bosonic action is nonchiral, so the obstruction comes from the fermion path integral. We subtract one trivial representation from \(\Delta_8-\rho_7\) to separate the internal finite-group anomaly from the gravitational anomaly of one right-moving Majorana fermion. The reduced virtual representation is
\begin{equation}
R_{\mathrm{int}}=\Delta_8|_\Gamma-\rho_7|_\Gamma-\mathbf1.
\label{eq:joyce-reduced-representation}
\end{equation}

The continuous Abelian analogy makes the calculation transparent. A 2D chiral anomaly is a signed sum of charge products over right-moving and left-moving fermions. For \(G=(\ZZ_2)^k\), the spin refinement is classified by \cite{GuoEtAl2018}
\begin{equation}
\Omega_3^{\mathrm{Spin}}(BG)
\cong
(\ZZ_8)^k
\oplus(\ZZ_4)^{\binom{k}{2}}
\oplus(\ZZ_2)^{\binom{k}{3}}.
\label{eq:finite-spin-bordism-group}
\end{equation}
A character is represented by charges \(q_a\in\FF_2\), where the generator \(g_a\) acts by \((-1)^{q_a}\). The three factors in \eqref{eq:finite-spin-bordism-group} are detected by the signed moments \cite{Grigoletto:2021zyv}
\begin{equation}
L_a=\sum_Rq_a-\sum_Lq_a,
\qquad
P_{ab}=\sum_Rq_aq_b-\sum_Lq_aq_b,
\qquad
T_{abc}=\sum_Rq_aq_bq_c-\sum_Lq_aq_bq_c.
\label{eq:joyce-anomaly-moments}
\end{equation}
They are evaluated modulo \(8\), \(4\), and \(2\), respectively.

Restricting \(\Delta_8\) and \(\rho_7\) to the four Joyce generators gives
\begin{equation}
L_a=0,
\qquad
P_{12}=-4,
\qquad
T_{123}=T_{124}=-2,
\label{eq:joyce-anomaly-moment-values}
\end{equation}
with every unlisted component equal to zero. The linear moments vanish modulo \(8\), the quadratic moments vanish modulo \(4\), and the cubic moments vanish modulo \(2\). The complete character lists and a basis-independent representation-ring derivation appear in \cref{app:finite-spin-anomaly}.

\begin{theorem}[Finite-spin anomaly and finite gauging]
\label{thm:joyce-finite-spin-anomaly}
For the Joyce action in \eqref{eq:joyce-action},
\begin{equation}
\mathcal A^{\mathrm{int}}_\Gamma(\Delta_8-\rho_7-\mathbf1)=0
\in
\operatorname{Hom}
\left(
\Omega_3^{\mathrm{Spin}}(B\Gamma),U(1)
\right).
\end{equation}
After choosing a trivialization of this internal anomaly, finite gauging defines a spin \(\mathcal N=(0,1)\) orbifold CFT with
\begin{equation}
(c_R,c_L)=\left(12,\frac{23}{2}\right),
\qquad
c_R-c_L=\frac12.
\end{equation}
\end{theorem}

The neutral representation is \(\Gamma\)-blind, so subtracting it does not change the internal finite-group anomaly. It only isolates the nonzero rank that controls the gravitational anomaly. Finite gauging preserves the local central charges, and only the combined system with its inverse 3D anomaly theory is invariant. To our knowledge, the finite-spin anomaly of this standalone seven-fermion Joyce orbifold has not been computed in the literature. The possible choices of anomaly trivialization and the full finite-group calculation are given in \cref{app:finite-spin-anomaly}.

Having established finite gauging, we turn to the geometry of the diagonal family.

\subsection{Cayley-preserving actions and Euler parity}
\label{subsec:cayley-signs}

The standard Joyce action has even Euler characteristic. We now ask whether different diagonal signs and half-shifts can produce an odd orbifold Euler characteristic. Let
\begin{equation}
g_{a,b}(x_i)=(-1)^{a_i}x_i+\frac{b_i}{2},
\qquad
a,b\in\FF_2^8,
\label{eq:half-affine-element}
\end{equation}
act on \(T^8=\RR^8/\ZZ^8\). The shifts do not act on \(dx_i\), so the Cayley form constrains only the sign vector \(a\). We denote the allowed signs by
\begin{equation}
C=\{a\in\FF_2^8\mid g_{a,0}^{*}\Phi=\Phi\}.
\label{eq:cayley-sign-subspace}
\end{equation}
This is a 4D binary sign subspace. Its nonzero elements are fourteen weight-four reflections and the full reflection
\begin{equation}
z=(1,1,1,1,1,1,1,1).
\end{equation}
The fourteen supports and a convenient basis of \(C\) are given in \cref{app:cayley-parity-proof}.

Consider an effective elementary abelian subgroup
\begin{equation}
\Gamma\leq C\oplus\FF_2^8
\end{equation}
of diagonal half-affine transformations. The same commuting-pair formula as \eqref{eq:joyce-standard-commuting-pairs} computes its unweighted orbifold Euler characteristic. The fixed-point mechanism can be read one coordinate at a time. If neither element reflects a coordinate, the common fixed set is either empty or contains a circle, so its Euler characteristic vanishes. If at least one element reflects the coordinate and the shifts are compatible, that coordinate contributes two common fixed points. A pair \(g_{a,b},g_{c,d}\) can therefore contribute only when
\begin{equation}
a\vee c=z,
\label{eq:full-coordinate-cover}
\end{equation}
where \(\vee\) denotes coordinatewise Boolean union. The shift compatibility conditions are linear equations over \(\FF_2\) and are displayed in \cref{app:cayley-parity-proof}.

Let \(N_\Gamma\) be the number of ordered commuting pairs that satisfy these conditions, and write
\begin{equation}
r=\dim_{\FF_2}\Gamma.
\end{equation}
Every contributing pair has \(2^8\) isolated common fixed points. Hence
\begin{equation}
\chi_{\mathrm{orb}}(T^8,\Gamma)=2^{8-r}N_\Gamma.
\label{eq:orbifold-euler-good-pairs}
\end{equation}
This formula turns the geometric question into a parity problem for compatible lifts.

\begin{theorem}[Diagonal Cayley half-affine parity]
\label{thm:diagonal-cayley-parity}
Let \(\Gamma\leq C\oplus\FF_2^8\) be an effective elementary abelian half-affine action on \(T^8\), where \(C\) is the complete diagonal sign stabilizer of the standard Cayley form. Then
\begin{equation}
\boxed{
\chi_{\mathrm{orb}}(T^8,\Gamma)\equiv0\pmod2.
}
\end{equation}
\end{theorem}

\begin{proof}
Let \(D\) be the image of the sign projection and let \(W\) be the pure-translation kernel. Write \(s=\dim D\), \(t=\dim W\), and \(r=s+t\). The Cayley sign structure implies that the sign span of a contributing pair is either the line \(\langle z\rangle\) or a plane \(\langle z,q\rangle\). If \(z\notin D\), then \(N_\Gamma=0\) and the result is immediate. We may therefore assume \(z\in D\). The line contributes \(3\cdot2^t\) ordered pairs. Each compatible plane contributes \(6\cdot2^{k_\Pi}\), where the affine lift equations give
\begin{equation}
k_\Pi\geq\max(0,2t-8).
\end{equation}
It follows that
\begin{equation}
N_\Gamma=3\cdot2^t+6\sum_\Pi2^{k_\Pi}.
\label{eq:good-pair-decomposition}
\end{equation}
If \(r\leq7\), the factor \(2^{8-r}\) in \eqref{eq:orbifold-euler-good-pairs} is even. If \(r\geq8\), the bound \(s\leq4\) gives \(t\geq r-4\), and the terms in \eqref{eq:good-pair-decomposition} are divisible by \(2^{r-7}\). Thus \(2^{8-r}N_\Gamma\) is again even. The complete lift calculation is given in \cref{app:cayley-parity-proof}.
\end{proof}

The Cayley condition is doing the work. Eight independent coordinate reflections do not preserve the Cayley form, and the same commuting-pair formula gives
\begin{equation}
\chi_{\mathrm{orb}}=3^8=6561.
\end{equation}
Thus orbifold averaging alone does not force evenness. To our knowledge, the parity theorem for the complete diagonal Cayley-preserving half-affine family has not appeared in the literature.

\paragraph{Full-sign lifts.}

The actions closest to Joyce's original construction have no pure translations and project isomorphically onto \(C\). After quotienting by coordinate quarter-period translations, their Euler characteristics take only four values.
\begin{proposition}[Full-sign half-shift classification]
\label{prop:full-sign-classification}
Suppose that the sign projection \(\Gamma\to C\) is an isomorphism. Modulo conjugation by coordinate quarter-period translations,
\begin{equation}
\boxed{
\chi_{\mathrm{orb}}\in\{48,144,336,720\}.
}
\end{equation}
\end{proposition}
The original Joyce action gives \(144\). The four multiplicities and the finite linear-algebra derivation are given in \cref{app:full-sign-counts}.

\paragraph{Smooth targets and larger finite families.}

Three Euler characteristics enter the geometric interpretation. The commuting-pair sum is an invariant of the orbifold action. Burnside's single sum gives the Euler characteristic of the coarse quotient, while a smooth resolution also depends on the local replacements and how they are assembled. For every diagonal half-affine action above, the coarse quotient Euler characteristic is divisible by \(16\). Under the classical Joyce assumptions on local singularities, exact stabilizers, and additive local Euler changes, a simultaneous resolution also has even Euler characteristic. The precise conditional statement and its proof are given in \cref{prop:joyce-resolution-parity}.

Permutations of the coordinate axes enlarge the finite symmetry group beyond the diagonal family. Every pure-linear cyclic subgroup, every elementary abelian subgroup of ranks one through four, and every pure-linear subgroup of order at most eight has even unweighted orbifold Euler characteristic. Including all half-lattice translations also leaves every cyclic subgroup even. These results are stated and derived in \cref{prop:pure-linear-monomial,prop:half-affine-cyclic}. They do not classify noncyclic affine groups, nonmonomial lattices, other translation denominators, or smooth resolutions.

The standard Joyce construction and its most direct torus-orbifold extensions therefore have even unweighted orbifold Euler characteristic. Under the resolution assumptions above, any corresponding smooth Spin(7) target has \(\mu_7=0\in\tmf_1\). The free antiholomorphic quotient in the next section leaves this framework and realizes the nonzero class.

\section{A free antiholomorphic quotient of the Fermat sextic}
\label{sec:fermat-sextic}

We now leave the diagonal Joyce family and construct a smooth torsion-free target with the nonzero geometric class in \(\tmf_1\). The target is a free antiholomorphic quotient of the Fermat sextic. We first derive \(\mu_7=\eta\), and then explain how the same class obstructs a compatible \(SU(4)\) reduction while leaving the chiral sigma model well defined.

\subsection{The Fermat quotient and the nonzero class}
\label{subsec:fermat-free-quotient}

Free antiholomorphic quotients of all-even-level Calabi--Yau fourfold models and their large-radius Betti numbers were analyzed in \cite[Sec.~3.2]{Blumenhagen:2001qx}. We specialize that construction to the Fermat sextic and compute its Spin(7) Smith class. The same antiholomorphic-quotient mechanism engineers 2D \(\mathcal N=(0,1)\) gauge theories on D1-branes probing the corresponding noncompact Spin(7) cones \cite{Franco:2021ixh,Franco:2021vxq}. Let
\begin{equation}
X_{\mathrm F}=
\left\{
z_0^6+z_1^6+z_2^6+z_3^6+z_4^6+z_5^6=0
\right\}
\subset\mathbb{CP}^5
\end{equation}
and let \(\sigma\) act by complex conjugation. The hypersurface is smooth because the six derivatives of its defining polynomial vanish simultaneously only at the zero vector.

A fixed point of \(\sigma\) can be represented by real homogeneous coordinates. It would obey
\begin{equation}
x_0^6+x_1^6+x_2^6+x_3^6+x_4^6+x_5^6=0.
\end{equation}
Every term is nonnegative, so the only real solution is the zero vector, which does not define a projective point. The Lefschetz hyperplane theorem implies that \(X_{\mathrm F}\) is simply connected. Complex conjugation therefore acts freely, and
\begin{equation}
Y_{\mathrm F}=X_{\mathrm F}/\langle\sigma\rangle
\end{equation}
is a smooth compact manifold with \(\pi_1(Y_{\mathrm F})=\ZZ_2\).

The quotient inherits its exceptional geometry directly from the Calabi--Yau cover. Let \(\omega\) be the Ricci-flat K\"ahler form in the hyperplane class and let \(\Omega\) be a parallel holomorphic four-form. Complex conjugation reverses the hyperplane class, so \(-\sigma^*\omega\) is another Ricci-flat K\"ahler form in the same class. Uniqueness gives \cite{Yau1978}
\begin{equation}
\sigma^*\omega=-\omega,
\qquad
\sigma^*\Omega=\overline\Omega,
\end{equation}
where a constant phase of \(\Omega\) has been chosen in the second relation. The combination
\begin{equation}
\Psi=\frac{1}{2}\omega^2+\operatorname{Re}\Omega
\label{eq:fermat-cayley-form}
\end{equation}
is invariant and descends to a closed Cayley form on \(Y_{\mathrm F}\). The universal cover has holonomy \(SU(4)\), while parallel transport around the nontrivial loop acts by complex conjugation. Thus
\begin{equation}
\Hol(Y_{\mathrm F})=SU(4)\rtimes\ZZ_2\subset\Spin(7).
\end{equation}
The quotient is a torsion-free barely Spin(7) manifold rather than a full-holonomy Spin(7) manifold \cite{Blumenhagen:2001qx,BraunSchaferNameki2018}.

It remains to determine which element of \(\tmf_1\) this target realizes. Let \(H\) denote the hyperplane class on \(X_{\mathrm F}\). The normal bundle sequence gives
\begin{equation}
c(TX_{\mathrm F})=\frac{(1+H)^6}{1+6H}
=1+15H^2-70H^3+435H^4.
\label{eq:fermat-total-chern}
\end{equation}
Since \(\int_{X_{\mathrm F}}H^4=6\),
\begin{equation}
\chi(X_{\mathrm F})=\int_{X_{\mathrm F}}c_4(TX_{\mathrm F})=435\cdot6=2610.
\label{eq:fermat-cover-euler}
\end{equation}
This value was computed in \cite{SethiVafaWitten1996}. The quotient is free and has degree two, so \(\chi(Y_{\mathrm F})=1305\). The Spin(7) Smith formula now gives
\begin{equation}
\boxed{
\begin{gathered}
\chi(Y_{\mathrm F})=1305,\\
\mu_7(Y_{\mathrm F})
=
(1305\bmod2)\eta
=
\eta
\in\tmf_1\cong\ZZ/2\{\eta\}.
\end{gathered}
}
\label{eq:fermat-mu7}
\end{equation}
The Fermat quotient therefore realizes the nonzero degree-one geometric class, in contrast with the smooth Joyce target in \eqref{eq:joyce-tmf-class}.

The complete normalized real-\(KO\) series carries the same answer at the cusp.
\begin{equation}
\Phi_2(Y_{\mathrm F},V_7;q)=\eta\in KO_1((q)).
\end{equation}
All positive powers of \(q\) vanish in this index through the oscillator cancellation of \cref{subsec:dirac-ramond-localization}. The class-level lift is \(\widetilde\eta\in\pi_1KO_{\mathrm{MF}}\), characterized by \(q\text{-exp}(\widetilde\eta)=\eta\). These are invariants of the smooth target bundle pair.

The example belongs to an open family. The Fermat polynomial is strictly positive on the unit sphere in \(\RR^6\), so every sufficiently small real perturbation still has empty real projective locus. Smoothness is also stable under small changes of the coefficients. These perturbations produce an open family of compact torsion-free barely Spin(7) quotients with the same odd Euler characteristic. Their invariants and deformation count are given in \cref{app:fermat-invariants}.

\subsection{The \texorpdfstring{\(SU(4)\)}{SU(4)} obstruction and the chiral sigma model}
\label{subsec:fermat-euler}

The nonzero class has a direct geometric meaning. For a fixed Spin(7) principal bundle \(P\),
\begin{equation}
\Spin(7)/SU(4)\cong S^6,
\end{equation}
and the associated \(S^6\)-bundle is the unit-sphere bundle of \(V_7\) \cite{Joyce2002Exceptional,Munoz:2013lqa}. A compatible \(SU(4)\) reduction is therefore equivalent to a nowhere-zero section of \(V_7\).

Pointwise, such a section selects a normalized two-form in \(\Lambda^2_7T^*Y_{\mathrm F}\) whose stabilizer inside \(\Spin(7)\) is \(SU(4)\). The global reduction exists only when these local choices patch over the whole target. The geometric meaning of the nonzero class can therefore be displayed as
\begin{equation}
\begin{aligned}
\mu_7(Y_{\mathrm F})=\eta
&\quad\Longrightarrow\quad
V_7\text{ has no nowhere-zero section},\\
&\quad\Longrightarrow\quad
P\text{ has no compatible }SU(4)\text{ reduction}.
\end{aligned}
\label{eq:fermat-su4-obstruction-chain}
\end{equation}

The K\"ahler form on the Calabi--Yau cover provides a local picture of the obstruction. It defines the expected section upstairs, but \(\sigma^*\omega=-\omega\), so it does not descend to \(Y_{\mathrm F}\). Equation \eqref{eq:fermat-su4-obstruction-chain} strengthens this failure of the natural section to a topological obstruction independent of the choice of section.

For simply connected targets, this mechanism gives a general criterion.
\begin{proposition}[Topological \texorpdfstring{\(SU(4)\)}{SU(4)} reduction]
\label{prop:su4-reduction-criterion}
Let \(X\) be a closed simply connected 8D Spin(7) manifold with fixed Spin(7) structure \(P\). Then
\begin{equation}
P\text{ admits a compatible topological }SU(4)\text{ reduction}
\quad\Longleftrightarrow\quad
\chi(X)\text{ is even}.
\label{eq:su4-reduction-criterion}
\end{equation}
\end{proposition}
The obstruction-theory proof is given in \cref{app:su4-reduction-proof}. The result concerns a topological reduction and does not imply an integrable complex structure or a parallel \(SU(4)\) structure. The quotient \(Y_{\mathrm F}\) is not simply connected, so the converse direction of \eqref{eq:su4-reduction-criterion} does not apply to it. Independently, its nonzero Smith class rules out a nowhere-zero section of \(V_7\).

The failure of an \(SU(4)\) reduction does not obstruct the chiral sigma model. In fact, the target itself is not String. Pullback to the Calabi--Yau cover gives
\begin{equation}
\lambda(TX_{\mathrm F})=-15H^2\neq0,
\label{eq:fermat-not-string}
\end{equation}
so \(TY_{\mathrm F}\) cannot admit a String structure. The fermion bundles satisfy
\begin{equation}
\lambda(TY_{\mathrm F})=\lambda(V_7).
\end{equation}
Together with the chosen universal nullhomotopy, this equality equips the virtual bundle \(TY_{\mathrm F}-V_7\) with the String structure required by the anomaly cancellation. The relevant object is the chiral bundle pair, not an ordinary String structure on the target.

This is the target-space form of the chiral cancellation in \cref{sec:canonical-model}. At the topological level, the chosen String structure trivializes the combined anomaly class. The quantum Pfaffian line uses the differential and \(B\)-field refinement described in \cref{subsec:virtual-string-structure}.

The same geometry also prevents an accidental tangent-bundle description. One finds
\begin{equation}
V_7\oplus\RR\not\cong TY_{\mathrm F}.
\end{equation}
Thus the canonical chiral bundle pair cannot be obtained by globally splitting a neutral line from the tangent Fermi bundle. The characteristic-class and Euler-class arguments are given in \cref{app:su4-reduction-proof}.

Adding one neutral left-moving Fermi multiplet would cancel the remaining gravitational degree, but it would not change this bundle obstruction. It produces an anomaly-free 2D theory without turning the canonical Fermi bundle into the tangent bundle.

We do not know a simply connected compact example with full Spin(7) holonomy and odd Euler characteristic.

\section{Discussion}
\label{sec:discussion}

In this paper, we construct the canonical Spin(7) sigma model, controlled by the following physical dictionary.
\begin{equation}
\begin{gathered}
\text{Spin(7) structure}
\longrightarrow
V_7=\Lambda^2_7T^*X,
\\
\lambda(TX)=\lambda(V_7)
\longrightarrow
\text{internal anomaly cancellation},
\\
TX-V_7\text{ String}
\longrightarrow
Z=s^{-1}(0),
\\
[Z]_{\tmf}
=
\bigl(\chi(X)\bmod2\bigr)\eta.
\end{gathered}
\label{eq:final-dictionary}
\end{equation}
The rank difference fixes the gravitational degree \(\nu=1\), while the String zero locus selects the class \((\chi(X)\bmod2)\eta\in\tmf_1\). The same locus controls the entire normalized real-\(KO\) Dirac--Ramond series and its degree-one class in \(KO_{\mathrm{MF}}\). The diagonal Joyce family has even orbifold Euler characteristic, and its Joyce resolutions remain even under the stated hypotheses. The Fermat quotient then realizes the nonzero class on a compact torsion-free target outside the diagonal family.

These results separate the geometric charge from its stronger field-theory interpretation. The construction is complete at the level of String bordism, \(\tmf\), the normalized \(KO\) series, and the class-level \(KO_{\mathrm{MF}}\) lift. The remaining questions concern the local field-theory representative, the interacting RG flow, and an exact presentation of the anomalous chiral theory. Its gravitational anomaly has invertible 3D inflow. Constructing a noninvertible bulk requires additional conformal-block and sewing data.

The results above lead to four concrete problems. The first three seek complementary field-theory realizations of the same degree-one class. The fourth asks whether compact full-holonomy Spin(7) geometry can realize its nonzero value.
\begin{itemize}[leftmargin=2em]
\item \textbf{Construct a local \(2|1\)-dimensional field-theory representative.}
Starting from a spin pair \((X,E)\) with a differential String trivialization of \(TX-E\), one should retain the tangent sigma-model anomaly and the Fermi Euler anomaly separately and cancel their internal parts only after tensoring. At the cusp, the resulting object should reproduce
\begin{equation}
p^{KO}_!\left[
e_{KO}(E)
\prod_{m\geq1}
\Sym_{q^m}(TX-\RR^d)
\Lambda_{-q^m}(E-\RR^r)
\right],
\end{equation}
and its Witten deformation should localize at object level to the sigma object on \(Z=s^{-1}(0)\). The product still takes values in the invertible gravitational anomaly line of degree \(\nu=d-r\). In degree one, differential forms alone cannot detect \(\eta\). The construction must retain the real \(KO\) module, anomaly-line descent, common spin structures in families, gluing, and reflection positivity.

\item \textbf{Identify the interacting sigma-model class.}
The canonical Lagrangian exists once the Spin(7) structure and differential anomaly trivialization are specified, and Equation~\eqref{eq:smith-main} computes the topological class attached to its bundle content. The remaining physical step is to show that the interacting sigma model represents this class in the field-theoretic construction of topological modular forms. The Joyce free theory gives a finite test. For a fixed internal anomaly trivialization, one can compute its spin-structure-resolved twisted torus amplitudes and mod-two genus and compare them with the vanishing geometric class at \(\chi=144\). The Fermat quotient poses the complementary problem because its geometric class is nonzero while no exact worldsheet presentation of the canonical chiral model is currently available.

\item \textbf{Extract the canonical model from the Spin(7) chiral algebra.}
The level-one branching
\begin{equation}
SO(8)_1\supset\Spin(7)_1\times\mathrm{Ising}
\end{equation}
identifies the central-charge difference needed to pass from eight to seven left-moving fermions \cite{ShatashviliVafa1994}. Since this is a conformal embedding rather than a tensor-product decomposition, the required operation is a spin fermionization with fixed-point resolution and compatible sewing. Completing it should identify the large-volume Fermi bundle with \(V_7\) and determine whether the chiral blocks admit a noninvertible 3D bulk. The gravitational anomaly alone supplies only invertible inflow. Relatedly, recent work suggests that exceptional-holonomy SCFTs may admit categorified structures over moduli spaces, such as stacks of fusion categories, as a noninvertible extension of classical Bagger--Witten-type data \cite{PerezLonaSharpeYu2025}.

\item \textbf{Find an odd-Euler full-holonomy target.}
The Fermat quotient proves that compact torsion-free Spin(7) geometry can carry \(\mu_7=\eta\), but its holonomy is \(SU(4)\rtimes\ZZ_2\). The strongest target is therefore a simply connected compact manifold with full Spin(7) holonomy and odd Euler characteristic. The diagonal Joyce theorem shows what must change. Candidates include noncyclic affine signed-monomial actions, nonmonomial integral lattices, higher-denominator translations, and resolutions beyond the classical Joyce local models. Each route has a concrete first test, which is to determine the Euler parity before attempting the full metric construction.
\end{itemize}

\section*{Acknowledgements}
\addcontentsline{toc}{section}{Acknowledgements}

The author thanks Yuji Tachikawa for valuable discussions, careful reading of the manuscript, and helpful comments. The author also thanks Arun Debray, Ying-Hsuan Lin, Brandon Rayhaun, Matthew Yu, and Hao Y. Zhang for fruitful discussions, from which the author learned some background on topological modular forms. The author also thanks Sebastian Franco, Alessandro Mininno, Alonso Perez-Lona, Eric Sharpe, and Angel Uranga for collaborations on related Spin(7) projects. The author also acknowledges the use of GPT-5.6 and Fable 5 for some computations and manuscript preparation.

\appendix
\crefalias{section}{appendix}
\crefalias{subsection}{subappendix}
\section{Conventions and low-degree identifications}
\label{app:conventions}

\paragraph{Topological modular forms.}

We use \(\tmf\) for connective topological modular forms, \(\Tmf\) for global sections over the compactified moduli of elliptic curves, and \(\TMF\) for the periodic theory over the smooth locus \cite{HillLawson2013}. The degree-one class used throughout the paper is \cite{TachikawaYamashita2021}
\begin{equation}
\pi_1\tmf\cong\ZZ/2\{\eta\}.
\end{equation}
The element \(\eta\) is the stable Hopf element. We write \(\eta_{\mathrm{Ded}}(\tau)\) for the Dedekind eta function.

The spectrum \(KO_{\mathrm{MF}}\) is the height-at-most-one homotopy pullback \eqref{eq:komf-pullback} introduced in \cite{BerwickEvans2023}, with cusp projection \(KO_{\mathrm{MF}}\to KO((q))\). In degree one, the cusp projection is an isomorphism on homotopy groups by Lemma~\ref{lem:pi1-komf}.

\paragraph{Chirality and degree.}

Let \(F_R\) and \(F_L\) denote the real bundles of right-moving and left-moving Majorana--Weyl fermions. Their contribution is
\begin{equation}
2(c_R-c_L)=\operatorname{rk}F_R-\operatorname{rk}F_L.
\end{equation}
For the canonical model,
\begin{equation}
F_R=TX,
\qquad
F_L=V_7,
\qquad
\nu=1.
\end{equation}
Changing the convention for which worldsheet light-cone direction is called right-moving reverses all chiral signs together and does not alter the mod-two class.

\paragraph{Characteristic classes in 8D.}

For an 8D spin manifold,
\begin{equation}
\widehat A_8=\frac{7p_1^2-4p_2}{5760},
\qquad
L_8=\frac{7p_2-p_1^2}{45}.
\end{equation}
For a real bundle \(V\),
\begin{equation}
\operatorname{ch}(V_{\mathbb C})
=\operatorname{rk}V+p_1(V)+\frac{p_1(V)^2-2p_2(V)}{12}+\cdots.
\end{equation}
Multiplying these expressions gives
\begin{equation}
\ind_{\mathbb C}(D_X\otimes TX_{\mathbb C})
=\left\langle\frac{37p_1^2-124p_2}{720},[X]\right\rangle.
\end{equation}
The Spin(7) relation \(4p_2-p_1^2=8e\) then gives \eqref{eq:index-ahat-euler-signature} \cite{Joyce2002Exceptional}.

\paragraph{The rank-seven Clifford convention.}

For a positive rank-seven Euclidean bundle, we stabilize by one positive real line and use the rank-eight ABS symbol. Restriction to the auxiliary line is followed by the graded Morita equivalence in \eqref{eq:graded-morita}. With this convention, the nonbounding spin circle represents the generator of \(KO_1\), and the AHR orientation maps it to \(\eta\in\tmf_1\). Any alternative Clifford sign convention changes intermediate module labels but leaves the mod-two class unchanged.

\section{Technical proofs and finite checks}
\label{app:computations}

This appendix records the finite-spin anomaly calculation, the Cayley lift counts, the local Euler changes for the classical Joyce resolutions, and the finite censuses used in the main text.

\paragraph{Finite-spin anomaly of the Joyce action.}
\label{app:finite-spin-anomaly}

Let \(\Gamma=\langle\alpha,\beta,\gamma,\delta\rangle\cong(\ZZ_2)^4\) act as in \eqref{eq:joyce-action}. We test the internal anomaly of the reduced rank-zero virtual representation
\begin{equation}
R_{\mathrm{int}}=\Delta_8|_\Gamma-\rho_7|_\Gamma-\mathbf1.
\end{equation}
The neutral representation removes the gravitational rank and carries no finite-group charge. We use the right-moving minus left-moving sign convention in \eqref{eq:joyce-anomaly-moments}.

Write a character of \(\Gamma\) as a four-vector in the ordered basis \((\alpha,\beta,\gamma,\delta)\). Restriction of the tangent representation gives
\begin{equation}
Q_R=
\{0100,0101,0110,0111,1000,1001,1010,1011\},
\label{eq:joyce-character-sets}
\end{equation}
while the seven left-moving characters are
\begin{equation}
Q_L=
\{0001,0010,0011,1100,1101,1110,1111\}.
\end{equation}
Substitution into \eqref{eq:joyce-anomaly-moments} gives the values in \eqref{eq:joyce-anomaly-moment-values}. They vanish in the corresponding \(\ZZ_8\), \(\ZZ_4\), and \(\ZZ_2\) factors of \eqref{eq:finite-spin-bordism-group}. The internal anomaly class is therefore trivial.

The same cancellation has a basis-independent expression in the representation ring.
\begin{equation}
\Delta_8|_\Gamma-\rho_7|_\Gamma
=
\mathbf1+\operatorname{Reg}_\Gamma
-2\operatorname{Reg}_{\Gamma/\langle\alpha\beta\rangle}.
\label{eq:joyce-representation-ring}
\end{equation}
Here \(\operatorname{Reg}_\Gamma\) contains every character of \(\Gamma\) once. The regular representation of \(\Gamma/\langle\alpha\beta\rangle\) is pulled back to \(\Gamma\), so it contains precisely the characters that are trivial on \(\alpha\beta\). Equation \eqref{eq:joyce-representation-ring} shows that the cancellation does not depend on the chosen generator basis.

The trivialization is not unique. One has
\begin{equation}
\operatorname{Hom}\bigl(\Omega_2^{\mathrm{Spin}}(B\Gamma),U(1)\bigr)
\cong(\ZZ_2)^{11},
\qquad
H^2(\Gamma,U(1))\cong(\ZZ_2)^6.
\end{equation}
The possible trivializations form a torsor for the first group. Removing the pure Arf factor leaves ten internal spin counterterms, only six of which arise from ordinary discrete torsion. Direct evaluation of the \(256\) flat backgrounds gives eleven distinct fermion boundary-sign profiles. These profiles are values of the fermion Pfaffian signs on backgrounds, not a list of trivializations or a construction of twisted amplitudes and sewing maps.

\paragraph{Cayley fixed sets and the diagonal parity proof.}
\label{app:cayley-parity-proof}

The diagonal half-affine action in \eqref{eq:half-affine-element} composes by addition in \(\FF_2^8\oplus\FF_2^8\). In standard coordinates, the Cayley form has fourteen monomial supports \cite{Joyce2002Exceptional}.
\begin{equation}
\begin{split}
\mathcal S_\Phi=\{&
1234,1256,1278,1357,1368,1458,1467,\\
&2358,2367,2457,2468,3456,3478,5678\}.
\end{split}
\end{equation}
A sign vector preserves \(\Phi\) precisely when
\begin{equation}
\sum_{i\in I}a_i=0\pmod2
\qquad
\text{for every }I\in\mathcal S_\Phi.
\label{eq:cayley-sign-constraints}
\end{equation}
The solution space is
\begin{equation}
C=
\left\langle
\begin{array}{c}
11110000\\
00001111\\
11001100\\
10101010
\end{array}
\right\rangle
\subset\FF_2^8,
\qquad
\dim_{\FF_2}C=4.
\label{eq:cayley-sign-generators}
\end{equation}
Its weight distribution is \(1+14y^4+y^8\). The unique weight-eight element is \(z=(1,1,1,1,1,1,1,1)\).

Let \(\Gamma\leq C\oplus\FF_2^8\) be effective. Its unweighted orbifold Euler characteristic is
\begin{equation}
\chi_{\mathrm{orb}}(T^8,\Gamma)
=
\frac1{|\Gamma|}
\sum_{g,h\in\Gamma}
\chi\bigl((T^8)^{g,h}\bigr).
\label{eq:orbifold-euler}
\end{equation}
For \(g_{a,b}\) and \(g_{c,d}\), a coordinate contributes only if at least one element reflects it. The two shifts must then obey
\begin{equation}
b_i=0\ \text{if }a_i=0,
\qquad
d_i=0\ \text{if }c_i=0,
\qquad
b_i=d_i\ \text{if }a_i=c_i=1.
\label{eq:affine-fixed-conditions}
\end{equation}
Together with \eqref{eq:full-coordinate-cover}, these conditions give two solutions in every coordinate and hence \(2^8\) isolated common fixed points. If \(N_\Gamma\) counts the ordered pairs obeying both conditions and \(r=\dim\Gamma\), division by \(|\Gamma|=2^r\) gives \eqref{eq:orbifold-euler-good-pairs}.

Let \(D\) be the image of the sign projection \(\Gamma\to C\), let \(W\) be its translation kernel, and write
\begin{equation}
s=\dim D,
\qquad
t=\dim W,
\qquad
r=s+t.
\end{equation}
The Cayley-preserving sign subspace has only weights zero, four, and eight. Therefore a pair satisfying \(a\vee c=z\) has sign span \(\langle z\rangle\) or \(\langle z,q\rangle\). When \(z\in D\), the \(2^t\) elements above \(z\) each determine the three ordered pairs \((1,g)\), \((g,1)\), and \((g,g)\). This gives the line contribution \(3\cdot2^t\).

If \(z\notin D\), then \(N_\Gamma=0\) and there is nothing to prove. We assume \(z\in D\) below.

For a plane \(\Pi=\langle z,q\rangle\), a contributing rank-two subgroup must project isomorphically to \(\Pi\). It is therefore a complement to \(W\). Choosing lifts of a basis gives two \(W\)-valued variables, while the coordinate conditions in \eqref{eq:affine-fixed-conditions} impose at most eight independent linear equations. The compatible complements, when nonempty, form an affine space of dimension
\begin{equation}
k_\Pi\geq\max(0,2t-8).
\end{equation}
Each complement has six ordered generating pairs. This proves \eqref{eq:good-pair-decomposition}.

If \(r\leq7\), the factor \(2^{8-r}\) in \eqref{eq:orbifold-euler-good-pairs} is even. If \(r\geq8\), then \(s\leq4\) gives \(t\geq r-4\). The line contribution is divisible by \(2^{r-7}\), and every plane contribution satisfies
\begin{equation}
v_2\bigl(6\cdot2^{k_\Pi}\bigr)
\geq 2t-7
\geq r-7.
\end{equation}
Thus \(N_\Gamma\) is divisible by \(2^{r-7}\), and \(2^{8-r}N_\Gamma\) is even.

The restriction to the Cayley sign subspace is essential. If the group is generated instead by eight independent coordinate reflections, the same calculation gives \(3^8=6561\). The theorem concerns the unweighted orbifold Euler characteristic of effective elementary abelian diagonal half-affine actions on the standard torus with half-lattice translations. It does not include weighted orbifold invariants or smooth resolutions.

Every subgroup in the theorem is specified by \(D\leq C\), \(W\leq\FF_2^8\), and a linear lift \(D\to\FF_2^8/W\). The number of such subgroups is
\begin{equation}
\sum_{s=0}^4\sum_{t=0}^8
\genfrac{[}{]}{0pt}{}{4}{s}_2
\genfrac{[}{]}{0pt}{}{8}{t}_2
2^{s(8-t)}
=488{,}176{,}700{,}923.
\end{equation}

\paragraph{Full-sign lifts.}
\label{app:full-sign-counts}

A full-sign lift has no pure translations, and its sign projection is an isomorphism onto \(C\). It is the graph of a linear map
\begin{equation}
B\colon C\longrightarrow\FF_2^8.
\label{eq:full-sign-graph}
\end{equation}
Conjugating the \(i\)-th coordinate by a quarter-period translation changes the shift assignment by
\begin{equation}
B(q)\longmapsto B(q)+q\odot u,
\qquad
u\in\FF_2^8,
\label{eq:quarter-translation-conjugation}
\end{equation}
where \(\odot\) is coordinatewise multiplication. The invariant normal form is
\begin{equation}
F_B(q)=B(q)+q\odot B(z).
\label{eq:full-sign-normal-form}
\end{equation}
It obeys \(F_B(z)=0\) and descends to
\begin{equation}
\overline F_B\colon C/\langle z\rangle\longrightarrow\FF_2^8.
\end{equation}
Thus quarter-translation classes are parametrized by \(8\times3\) binary matrices.

Let \(d=\dim\ker\overline F_B\). The line \(\langle z\rangle\) contributes three ordered pairs, while every nonzero kernel vector selects a compatible plane containing \(z\) and contributes six ordered pairs. Since \(|\Gamma|=16\),
\begin{equation}
\chi_{\mathrm{orb}}
=16\left[3+6(2^d-1)\right].
\label{eq:full-sign-euler}
\end{equation}
The four kernel dimensions \(d=0,1,2,3\) give \(48,144,336,720\). The number of \(8\times3\) matrices of rank \(m\) is
\begin{equation}
N_{8,3}(m)
=
\prod_{i=0}^{m-1}
\frac{(2^8-2^i)(2^3-2^i)}{2^m-2^i}.
\end{equation}
Using \(m=3-d\), the respective numbers of quarter-translation conjugacy classes are
\begin{equation}
16{,}322{,}040,
\qquad
453{,}390,
\qquad
1{,}785,
\qquad
1.
\end{equation}
This quotient is only by coordinate quarter-period translations. It is not a classification under arbitrary affine or lattice conjugation.

\paragraph{Coarse quotients and classical Joyce resolutions.}
\label{app:joyce-resolution-proof}

The commuting-pair invariant \(\chi_{\mathrm{orb}}\), the Euler characteristic of the coarse quotient, and the Euler characteristic of a smooth resolution are different quantities. Burnside's formula gives
\begin{equation}
\chi(T^8/\Gamma)
=
\frac1{|\Gamma|}\sum_{g\in\Gamma}\chi\bigl((T^8)^g\bigr).
\end{equation}
Let \(s\) be the dimension of the sign image and \(t\) the dimension of the translation kernel. A nonzero contribution occurs only above the full reflection \(z\). There are \(2^t\) such elements, each with \(2^8\) fixed points. Hence
\begin{equation}
\chi(T^8/\Gamma)
=
\begin{cases}
0,&z\notin\operatorname{im}(\Gamma\to C),\\[2mm]
2^{8-s},&z\in\operatorname{im}(\Gamma\to C).
\end{cases}
\label{eq:coarse-quotient-euler}
\end{equation}

\begin{proposition}[Coarse quotient and Joyce resolutions]
\label{prop:joyce-resolution-parity}
For every effective \(\Gamma\leq C\oplus\FF_2^8\), the coarse quotient Euler characteristic in \eqref{eq:coarse-quotient-euler} is divisible by \(16\).

Suppose in addition that every singular chart is one of the classical Joyce types I through V, every point in a compatible rank-two common fixed set has that exact stabilizer and no larger isotropy, and the Euler characteristic of the simultaneous resolution equals the coarse quotient Euler characteristic plus the sum of the classical local Euler changes. Then the Euler characteristic of the smooth resolution is even.
\end{proposition}

For classical Joyce singularities, the local Euler changes for types I through V are
\begin{equation}
\begin{array}{c|ccccc}
\text{type}&\mathrm{I}&\mathrm{II}&\mathrm{III}&\mathrm{IV}&\mathrm{V}\\ \hline
\Delta\chi&0&8&1&0&4
\end{array}
\end{equation}
for either allowed resolution choice when a type splits \cite{Joyce1996Spin7,Taylor1996}. Only type III can change parity. For a compatible plane \(\langle z,q\rangle\), the two weight-four elements \(q\) and \(z+q\) reflect complementary 4-planes. The local quotient is therefore \((\RR^4/\{\pm1\})\times(\RR^4/\{\pm1\})\). An exact rank-two type-III stabilizer has \(256\) common fixed points. If \(r=\dim\Gamma\), the residual free action produces \(2^{10-r}\) quotient components, which is even for \(r\leq9\). At \(r=10\), the bound \(s\leq4\) gives \(t\geq6\), and every nonempty family of compatible complements has dimension at least \(2t-8\geq4\). Its size is therefore divisible by \(16\). For \(r>10\), a group of order \(2^{r-2}>256\) cannot act freely on the fixed set. The total type-III contribution is even under the hypotheses of Proposition~\ref{prop:joyce-resolution-parity}.

This second statement is conditional and does not establish the existence of a simultaneous resolution. Higher isotropy, isolated \(\RR^8/\{\pm1\}\) singularities, nonclassical local models, and nonadditive resolution topology require a separate analysis.

\paragraph{Fermat quotient invariants and deformations.}
\label{app:fermat-invariants}

The Calabi--Yau cover in \cref{subsec:fermat-free-quotient} has \(\widehat A(X_{\mathrm F})=\chi(\mathcal O_{X_{\mathrm F}})=2\). The \(\widehat A\) characteristic number scales with the degree of the free double cover, so
\begin{equation}
\widehat A(Y_{\mathrm F})=1.
\end{equation}
Together with \(3\Sign(Y)=\chi(Y)+48\widehat A(Y)\), this gives
\begin{equation}
\Sign(Y_{\mathrm F})=451.
\end{equation}
The hyperplane class is odd under complex conjugation, while \(H^3(X_{\mathrm F};\mathbb Q)=0\). Hence
\begin{equation}
b_2(Y_{\mathrm F})=b_3(Y_{\mathrm F})=0.
\end{equation}
Poincar\'e duality, the Euler characteristic, and the signature then give
\begin{equation}
b_4^+(Y_{\mathrm F})=877,
\qquad
b_4^-(Y_{\mathrm F})=426.
\end{equation}
These values agree with the large-radius quotient formulas of \cite{Blumenhagen:2001qx}.

The odd-Euler example persists under real deformations. There are \(\binom{11}{5}=462\) degree-six monomials in six variables. Removing the overall polynomial scale and the 35D projective coordinate freedom leaves
\begin{equation}
462-1-35=426
\end{equation}
real invariant complex-structure deformations. Including the volume gives \(427\) metric deformations, in agreement with the torsion-free Spin(7) count \cite{Joyce1999Spin7,Joyce2002Exceptional,BraunSchaferNameki2018}
\begin{equation}
b_4^-(Y_{\mathrm F})+1=427.
\end{equation}

\paragraph{The obstruction to an \texorpdfstring{\(SU(4)\)}{SU(4)} reduction.}
\label{app:su4-reduction-proof}

We prove \cref{prop:su4-reduction-criterion}. The fiber of the reduction problem is \(\Spin(7)/SU(4)\cong S^6\), so a compatible reduction is a section of the associated unit-sphere bundle of \(V_7\). Standard obstruction theory gives a primary obstruction in \(H^7(X;\ZZ)\) because \(\pi_6(S^6)=\ZZ\) \cite{CadekCrabbVanzura2007}. It vanishes when \(X\) is simply connected by Poincar\'e duality. The only remaining obstruction lies in
\begin{equation}
H^8(X;\pi_7(S^6))
\cong
H^8(X;\ZZ/2).
\end{equation}
Changing the section over the seven-skeleton can shift this class by \(Sq^2H^6(X;\ZZ/2)\). On an 8D spin manifold, the Wu formula gives
\begin{equation}
\left\langle Sq^2a,[X]\right\rangle
=
\left\langle w_2(X)a,[X]\right\rangle
=0,
\end{equation}
so this top-degree indeterminacy vanishes. Stabilization \(\pi_7(S^6)\to\pi_1^S\) is an isomorphism and identifies the remaining bit with the stable zero-locus class \cite{DebrayEtAl2026Smith}. The Smith formula therefore gives
\begin{equation}
P\text{ reduces to }SU(4)
\quad\Longleftrightarrow\quad
\mu_7(X)=0
\quad\Longleftrightarrow\quad
\chi(X)\text{ is even}.
\end{equation}
This is a topological reduction criterion for a fixed Spin(7) structure. It is compatible with the related 8D vector-bundle obstruction calculations of \cite{CadekCrabbVanzura2007}.

For the Fermat quotient, the tangent bundle itself is not String. The Chern-class calculation in \eqref{eq:fermat-total-chern} gives
\begin{equation}
c_1(TX_{\mathrm F})=0,
\qquad
c_2(TX_{\mathrm F})=15H^2,
\end{equation}
and hence
\begin{equation}
p_1(TX_{\mathrm F})=-30H^2,
\qquad
\lambda(TX_{\mathrm F})=-15H^2\neq0.
\end{equation}
If \(TY_{\mathrm F}\) admitted a String structure, its pullback to \(X_{\mathrm F}\) would contradict this result. By contrast, the universal Spin(7) representation identity gives \(\lambda(TY_{\mathrm F})=\lambda(V_7)\), and the chosen nullhomotopy equips \(TY_{\mathrm F}-V_7\) with its String structure.

Finally, \(V_7\oplus\RR\) has a nowhere-zero section from its trivial summand. The nonzero Euler number
\begin{equation}
\left\langle e(TY_{\mathrm F}),[Y_{\mathrm F}]\right\rangle=1305
\end{equation}
prevents \(TY_{\mathrm F}\) from having such a section. Therefore \(V_7\oplus\RR\not\cong TY_{\mathrm F}\).

\paragraph{Spin(7) representation and index checks.}

One check verifies the maximal-torus weight sums for \(\rho_7\) and \(\Delta_8\), together with the restricted Joyce characters. A second expands
\begin{equation}
\widehat A(TX)\operatorname{ch}(TX_{\mathbb C})
\end{equation}
through degree eight, imposes \(4p_2-p_1^2=8e\), and reproduces \eqref{eq:index-ahat-euler-signature}. A third reproduces the Fermat Chern-class and quotient invariants.

\paragraph{Diagonal half-affine actions.}

A direct enumeration constructs the Cayley-preserving sign subspace, verifies all coordinate fixed-set rules, covers an independent \(2^{16}\) slice of full-sign lifts, and reproduces the four values in Proposition~\ref{prop:full-sign-classification}. For the slice with the first two generator shifts fixed to zero, the distribution is
\begin{equation}
\begin{array}{c|rrr}
\chi_{\mathrm{orb}}&144&336&720\\ \hline
\#\text{ lifts}&64{,}770&765&1.
\end{array}
\end{equation}
The same fixed-set rules give \(6561\) for eight independent coordinate reflections after the Cayley condition is removed.

\paragraph{Pure-linear signed-monomial groups.}
\label{app:signed-monomial-checks}

Let
\begin{equation}
M=\{A\in O(8,\ZZ)\mid A^*\Phi=\Phi\}.
\end{equation}
This is the complete signed-monomial stabilizer of the standard Cayley form. For a finite subgroup \(H\leq M\), including a nonabelian subgroup, the unweighted orbifold Euler characteristic is
\begin{equation}
\chi_{\mathrm{orb}}(T^8,H)
=
\frac1{|H|}
\sum_{\substack{g,h\in H\\gh=hg}}
\chi\bigl((T^8)^{g,h}\bigr).
\end{equation}

\begin{proposition}[Pure-linear signed-monomial parity]
\label{prop:pure-linear-monomial}
For the action of \(M\) on \(T^8\), every cyclic subgroup, every elementary abelian two-subgroup of ranks one through four, and every subgroup of order at most eight has even unweighted orbifold Euler characteristic.
\end{proposition}

The signed-monomial census finds \(1344\) Cayley-support permutations and sixteen exact sign lifts of each, giving \(|M|=21{,}504\). The elementary abelian subgroup counts are
\begin{equation}
\begin{array}{c|rrrr}
\text{rank}&1&2&3&4\\ \hline
\#\text{ subgroups}&575&4{,}795&3{,}047&240.
\end{array}
\end{equation}
The complete cyclic census contains \(7{,}200\) subgroups, and every one has even unweighted orbifold Euler characteristic.
The complete order-at-most-eight census is
\begin{equation}
\begin{array}{c|r@{\qquad}c|r}
\text{type}&\#&\text{type}&\#\\ \hline
C_1&1&C_2&575\\
C_3&448&C_4&1{,}120\\
V_4&4{,}795&C_6&2{,}240\\
S_3&2{,}688&C_7&512\\
C_8&1{,}344&C_4\times C_2&2{,}940\\
C_2^3&3{,}047&D_8&10{,}780\\
Q_8&644&&
\end{array}
\end{equation}
Here \(D_8\) denotes the dihedral group of order eight. Every listed subgroup has even unweighted orbifold Euler characteristic.

The computation enumerates subgroups by closure from generating tuples, removes repeated subgroups by their exact element sets, and checks the resulting group orders and element-order distributions. It contains all \(7{,}200\) cyclic subgroups, all elementary abelian subgroups of ranks one through four, and all \(31{,}134\) subgroups of order at most eight. Common fixed sets are computed by exact integer and finite-field linear algebra.

\paragraph{Half-affine cyclic groups.}

Set
\begin{equation}
\mathcal A=T^8[2]\rtimes M,
\qquad
|\mathcal A|=5{,}505{,}024.
\end{equation}
\begin{proposition}[Half-affine cyclic parity]
\label{prop:half-affine-cyclic}
Every cyclic subgroup of \(\mathcal A\) has even unweighted orbifold Euler characteristic.
\end{proposition}

For an element \((A,u)\) with linear part of order \(m\),
\begin{equation}
(A,u)^m=(1,N_Au),
\qquad
N_A=1+A+\cdots+A^{m-1}
\end{equation}
over \(\FF_2\). The rank and kernel of \(N_A\) determine whether the affine element has order \(m\) or \(2m\). Dividing the resulting element counts by Euler's totient \(\varphi(n)\) gives the number of cyclic subgroups of order \(n\).

The norm-rank census for \(\mathcal A\) gives the cyclic subgroup distribution
\begin{equation}
\begin{array}{c|rrrrrrrrr}
|H|&1&2&3&4&6&7&8&12&14\\ \hline
\#H&
1&29{,}183&7{,}168&216{,}832&508{,}928&
32{,}768&408{,}576&200{,}704&229{,}376.
\end{array}
\end{equation}
The Euler distribution is
\begin{equation}
\begin{array}{c|rrrrrr}
\chi_{\mathrm{orb}}
&0&96&144&192&240&384\\ \hline
\#H
&1{,}018{,}368&174{,}080&286{,}720&
125{,}440&28{,}672&256.
\end{array}
\end{equation}
The counts sum to \(1{,}633{,}536\). An offset-sensitive affine signed-graph check covers all \(256\) shifts for a representative of every norm profile.

The pure-linear proposition covers exactly the subgroup families stated there, while the half-affine proposition covers cyclic subgroups only. Neither result classifies noncyclic affine subgroups, nonmonomial lattices, other translation denominators, weighted orbifold Euler characteristics, or smooth resolutions.

\paragraph{Worldsheet checks.}
\label{app:worldsheet-checks}

The representation-ring relation, the signed anomaly moments, and the eleven boundary-sign profiles used in \cref{subsec:joyce-finite-spin} are verified by the same exact evaluation of the \(256\) flat backgrounds. Further checks reproduce the even-spin genus \eqref{eq:joyce-even-spin-genus} and its \(q\)-expansion, verify the theta-subgroup conjugacy, and confirm the central charges, conformal weights, modular data, inverse twists, and character identity of the diagonal \(E_{8,1}\) extension in \cref{app:joyce-worldsheet-comparisons}.

\section{Worldsheet constructions at the Joyce point}
\label{app:joyce-worldsheet-comparisons}

The standard Joyce action also appears in two exact worldsheet constructions that are distinct from the standalone chiral CFT in \cref{subsec:joyce-canonical-orbifold}. We record their full character formulas and protected genera here.

\paragraph{The heterotic conformal embedding.}

Sugiyama and Yamaguchi constructed a full heterotic, GSO-summed genus-one orbifold for the same Joyce action \cite{Sugiyama:2001qh}. Their sector blocks have the form
\begin{equation}
Z^{\mathrm{het}}_{g,h}
=
Z^B_{g,h}Z^{\mathrm{RNS}}_{g,h}
\overline{
\chi^{E_8^{(1)}}_{g,h}\chi^{E_8^{(2)}}_0
},
\end{equation}
and the assembled sum is modular invariant. This is an exact genus-one heterotic construction built from the same affine \(\mathfrak{so}(7)_1\) sector. It does not by itself give the standalone seven-fermion orbifold in \eqref{eq:joyce-canonical-orbifold}.

For the conformal embedding in \eqref{eq:joyce-heterotic-conformal-embedding}, the vacuum character branches as
\begin{equation}
\boxed{
\chi^{E_8}_0
=
\chi^7_0\chi^9_0
+\chi^7_v\chi^9_v
+\chi^7_s\chi^9_s.
}
\label{eq:e8-diagonal-branching}
\end{equation}
The three branches transform as
\begin{equation}
(\mathbf1,\mathbf1),
\qquad
(\mathbf7,\mathbf9),
\qquad
(\mathbf8,\mathbf{16}),
\end{equation}
with paired conformal weights
\begin{equation}
(0,0),
\qquad
\left(\frac12,\frac12\right),
\qquad
\left(\frac7{16},\frac9{16}\right).
\end{equation}
Every branch therefore has integral total conformal weight. The central charges also add to eight.
\begin{equation}
c\bigl(\mathfrak{so}(7)_1\bigr)
+c\bigl(\mathfrak{so}(9)_1\bigr)
=\frac72+\frac92=8.
\end{equation}

The diagonal pairing prevents the physical Hilbert space from factorizing into seven-fermion and nine-fermion sectors. In particular, the spinor weights \(7/16\) and \(9/16\) add to one and define a local sector only together. The modular-invariant heterotic partition function therefore cannot be divided by an \(\mathfrak{so}(9)_1\) character to recover the standalone seven-fermion theory.

At the affine-character level, the same extension appears in
\begin{equation}
O_7O_9+V_7V_9+S_7S_9
=
\frac12
\left[
\left(\frac{\theta_3}{\eta_{\mathrm{Ded}}}\right)^8
+\left(\frac{\theta_4}{\eta_{\mathrm{Ded}}}\right)^8
+\left(\frac{\theta_2}{\eta_{\mathrm{Ded}}}\right)^8
\right]
=\chi^{E_8}_0.
\end{equation}
The categorical and theta-function checks are collected in \cref{app:worldsheet-checks}. An explicit standalone construction would instead require its own spin-structure-resolved twisted amplitudes and sewing data.

\paragraph{The ordinary \texorpdfstring{\(\mathcal N=(1,1)\)}{N=(1,1)} orbifold.}

The familiar \(\mathcal N=(1,1)\) Joyce orbifold answers a different protected question. Both chiralities use the tangent representation, so it has
\begin{equation}
c_L=c_R=12
\end{equation}
and no gravitational anomaly. Its fully periodic protected index is known exactly \cite{ShatashviliVafa1994,BenjaminEtAl2014}.
\begin{equation}
I_{\mathrm{RR}}=\chi_{\mathrm{orb}}(T^8,\Gamma)=144.
\end{equation}
The even-spin protected genus contains more information \cite{ShatashviliVafa1994,BenjaminEtAl2014}.
\begin{equation}
Z_{\mathrm{NS},+}(\tau)
=
\operatorname{Tr}_{\mathrm{NS},\mathrm{R}}
(-1)^{F_R}
q^{L_0-c/24}
\bar q^{\bar L_0-c/24}.
\end{equation}
For a compact full-holonomy Spin(7) target with \(\widehat A=1\),
\begin{equation}
Z_{\mathrm{NS},+}(\tau)
=
K(\tau)+8-\frac{\chi}{3},
\end{equation}
where
\begin{equation}
K(\tau)
=
\frac{\Delta(\tau)^2}{\Delta(2\tau)\Delta(\tau/2)}
-24
=q^{-1/2}+276q^{1/2}+2048q+\cdots.
\end{equation}
For the first Joyce example,
\begin{equation}
\boxed{
Z_{\mathrm{NS},+}=K-40.
}
\label{eq:joyce-even-spin-genus}
\end{equation}
At the orbifold point this follows from
\begin{equation}
16\left(
\frac{\theta_2^4}{\theta_4^4}
+\frac{\theta_4^4}{\theta_2^4}
-2
\right)
=K-40.
\end{equation}

The even-spin genus transforms under
\begin{equation}
\Gamma_\theta=\langle S,T^2\rangle,
\end{equation}
which is conjugate to \(\Gamma_0(2)\). This modularity is associated with the level-two spectrum \(\Tmf_0(2)\) over \(\ZZ[1/2]\) \cite{HillLawson2013}. The protected genus alone does not define a homotopy class in that spectrum.

The exact free-field presentation applies at the singular flat orbifold point. A smooth Joyce resolution generally requires an \(\alpha'\)-corrected background, although supersymmetry can persist after those corrections are included \cite{BeckerRobbinsWitten2014}. The corresponding smooth target has even Euler characteristic and hence vanishing geometric Smith charge. These constructions establish the neighboring exact worldsheet theories, but neither computes a protected index of the standalone chiral CFT.

\section{Suspension symbol and the \(\eta\)-factor}
\label{app:suspension-eta}

In the proof of Proposition~\ref{prop:odd-rank-abs}, let \(E\to X\) be a rank-seven
spin bundle and \(W=E\oplus\RR\). Let
\[
a_W\in KO_c^0(W),\qquad
U_W:=\beta^{-1}a_W\in KO_c^8(W)
\]
be, respectively, the stabilized rank-eight ABS symbol and its normalized Thom class.

Let \(u_1\in KO_c^1(\RR)\cong \widetilde{KO}^1(S^1)\cong KO^0(\mathrm{pt})\cong \ZZ\)
be the Thom class of the positively oriented spin line.
Let
\[
\delta_t:=\bigl[\RR\to\RR,\, t\,\mathrm{id},\,\RR\bigr]\in KO_c^0(\RR)
\cong \widetilde{KO}^0(S^1)\cong KO^{-1}(\mathrm{pt})\cong \ZZ/2
\]
be the compactly supported degree-zero suspension symbol.
After one-point compactifying \(\RR\to S^1\), \(\delta_t\) has normalized values
\(-1\) and \(+1\) at the two ends, so its clutching class is the Möbius line
minus the trivial line. Hence \(\delta_t\) is the nontrivial element of
\(\widetilde{KO}(S^1)\), and therefore
\begin{equation}
\delta_t=\eta\,u_1.
\label{eq:eta-delta-relation}
\end{equation}

Consider the inclusion \(j:\RR\to W\), \(j(t)=(0,t)\). Thom multiplicativity
implies
\begin{equation}
j^*U_W=e_{KO}(E)\boxtimes u_1.
\label{eq:susp-thom-mult}
\end{equation}
On the other hand, using compatible \(\Spin(7)\subset\Spin(8)\) spinor choices,
the \(W\)-symbol \(a_W\) restricted to \(\RR\) is exactly the product of the
zero-mode bundle with \(\delta_t\):
\begin{equation}
j^*a_W=[S(E)]\boxtimes\delta_t.
\label{eq:susp-restriction-a}
\end{equation}
Multiplying \eqref{eq:susp-restriction-a} by \(\beta^{-1}\) and using
\eqref{eq:eta-delta-relation} gives
\begin{equation}
j^*U_W=\beta^{-1}[S(E)]\boxtimes\delta_t
\;=\;\eta\,\beta^{-1}[S(E)]\boxtimes u_1.
\label{eq:susp-restriction-U}
\end{equation}
Comparing \eqref{eq:susp-thom-mult} with \eqref{eq:susp-restriction-U} and
canceling the \(u_1\)-factor by inverse Thom isomorphism on the auxiliary line
yields
\begin{equation}
e_{KO}(E)=\eta\,\beta^{-1}[S(E)].
\label{eq:susp-result}
\end{equation}
Because \(S(E)\) has real rank eight and \(8\eta=0\), this is equivalent to
\[
e_{KO}(E)=\eta\,\beta^{-1}\bigl([S(E)]-8\bigr).
\]

\printbibliography[heading=bibliography]

\end{document}